\documentclass[11pt,twoside]{article}

\usepackage{geometry}                
\usepackage{graphicx}
\usepackage{amsmath,amssymb, amsthm,epsfig,bm}
\usepackage{setspace}
\usepackage{epstopdf}
\usepackage{psfrag}
\usepackage{color,soul}
\usepackage{verbatim}
\usepackage{algorithm,algorithmic}
\usepackage{multirow}
\usepackage{natbib}
\usepackage{xr}
\usepackage{mathrsfs}
\usepackage{cases}
\usepackage[colorlinks=true,citecolor=blue,linkcolor=blue]{hyperref}
\usepackage{enumitem}

\newcommand{\titleshort}{{Estimator Augmentation}}
\newcommand{\authorshort}{{Zhou and Min}}
\newcommand{\apptitle}{Appendix}  

\setlist{noitemsep, topsep=0pt}
\bibpunct{(}{)}{;}{a}{}{,} 
\usepackage[english]{babel}
\usepackage[utf8]{inputenc}
\usepackage{fancyhdr}
 
\numberwithin{equation}{section}
\theoremstyle{theorem}
\newtheorem{theorem}{Theorem}[section]
\newtheorem{lemma}[theorem]{Lemma}
\newtheorem{proposition}[theorem]{Proposition}
\newtheorem{corollary}[theorem]{Corollary}

\theoremstyle{definition}
\newtheorem{definition}{Definition}
\newtheorem{remark}{Remark}
\newtheorem{assumption}{Assumption}
\newtheorem{example}{Example}
\newtheorem{routine}{Algorithm}

\newcommand{\myappendix}{\section*{\hfil \apptitle \hfil}\appendix}

\newcommand{\bfI}{\mathbf{I}}

\newcommand{\bfzr}{\mathbf{0}}
\newcommand{\calK}{\mathcal{K}}
\newcommand{\calA}{\mathcal{A}}

\newcommand{\tdH}{\tilde{H}}

\newcommand{\tdv}{\tilde{v}}

\newcommand{\calB}{\mathcal{B}}
\newcommand{\calC}{\mathcal{C}}

\newcommand{\calI}{\mathcal{I}}
\newcommand{\calM}{\mathcal{M}}

\newcommand{\calG}{\mathcal{G}}

\newcommand{\calV}{\mathcal{V}}
\newcommand{\scrA}{\mathscr{A}}

\newcommand{\scrK}{\mathscr{K}}

\newcommand{\mbB}{\mathbb{B}}
\newcommand{\mbD}{\mathbb{D}}
\newcommand{\R}{\mathbb{R}}
\newcommand{\N}{\mathbb{N}}
\newcommand{\mbS}{\mathbb{S}}
\newcommand{\Prob}{\mathbb{P}}

\newcommand{\alp}{\alpha}

\newcommand{\hbeta}{\hat{\beta}}

\newcommand{\tdbeta}{\tilde{\beta}}

\newcommand{\veps}{\varepsilon}

\newcommand{\hTheta}{\hat{\Theta}}

\newcommand{\hsigma}{\hat{\sigma}}

\newcommand{\hgamma}{\hat{\gamma}}

\newcommand{\lmd}{\lambda}

\newcommand{\Omg}{\Omega}

\newcommand{\hzeta}{\hat{\zeta}}
\newcommand{\hb}{\hat{b}}

\newcommand{\hy}{\hat{y}}

\newcommand{\E}{\mathbb{E}}

\newcommand{\calE}{\mathcal{E}}
\newcommand{\I}{\bm{1}}
\newcommand{\defi}{\mathop{=}\limits^{\Delta}}  
\newcommand{\toas}{\mathop{\longrightarrow}\limits^{a.s.}}  
\newcommand{\eqinL}{\mathop{=}\limits^{d}}      
\newcommand{\dnorm}{\mathcal{N}}

\newcommand{\trans}{\mathsf{T}}
\newcommand{\supth}{\text{th}}

\newcommand{\sqn}{\sqrt{n}}
\newcommand{\sqm}{\sqrt{m}}
\newcommand{\frn}{\frac{1}{n}}

\newcommand{\vn}{\varnothing}
\newcommand{\la}{\langle}
\newcommand{\ra}{\rangle}
\newcommand{\grad}{\nabla}

\DeclareMathOperator*{\argmin}{argmin}
\DeclareMathOperator*{\argmax}{argmax}

\DeclareMathOperator{\diag}{diag}
\DeclareMathOperator{\sgn}{sgn}
\DeclareMathOperator{\row}{row}
\DeclareMathOperator{\nul}{null}
\DeclareMathOperator{\rank}{rank}

\DeclareMathOperator{\spn}{span}

\RequirePackage{tikz}

\newcommand{\estpivot}{\hat{b}}

\newcommand{\siglevel}{\delta}

\begin{document}

\title{Estimator Augmentation with Applications \\ in High-Dimensional Group Inference}
\author{Qing Zhou\thanks{UCLA Department of Statistics}
\thanks{(Email: zhou@stat.ucla.edu) Supported by NSF grants DMS-1055286 and IIS-1546098.}
 \and Seunghyun Min$^*$
}
\maketitle

\begin{abstract}
To make inference about a group of parameters on high-dimensional data,
we develop the method of estimator augmentation for the block Lasso, which is defined
via the block norm. 
By augmenting a block Lasso estimator $\hbeta$ with the subgradient $S$ of the block norm evaluated
at $\hbeta$, we derive a closed-form density for the joint distribution of $(\hbeta,S)$ under a high-dimensional
setting. This allows us to draw from an estimated sampling distribution of $\hbeta$, 
or more generally any function of $(\hbeta,S)$, by Monte Carlo algorithms.
We demonstrate the application of estimator augmentation in group inference 
with the group Lasso and a de-biased group Lasso constructed as a function of $(\hbeta,S)$. 
Our numerical results show that
importance sampling via estimator augmentation can be orders of magnitude more efficient 
than parametric bootstrap in estimating tail probabilities for significance tests.
This work also brings new insights into the geometry of the sample space
and the solution uniqueness of the block Lasso.

\end{abstract}

\section{Introduction}

There has been a fast growth of high-dimensional data in many areas,
such as genomics and the social science. 
Statistical inference for high-dimensional models becomes a necessary tool
for scientific discoveries from such data. For example, significance
tests have been performed to screen millions of genomic loci for disease markers.
These applications have motivated the recent development in high-dimensional statistical inference.
Some methods make use of sample splitting and subsampling
to quantify estimation errors and significance \citep{Wasserman09,Meinshausen09,Meinshausen10},
while others rely on the bootstrap to approximate the sampling distributions
of Lasso-type estimators \citep{Chatterjee13,Zhou14}.
For Gaussian linear models, an interesting idea of de-biasing the Lasso \citep{Tibshirani96} 
has been developed by a few groups 
\citep{ZhangZhang11,vandeGeer13,Javanmard13b}. 
In addition, there are various other inferential methods 
\citep{Lockhart13,Lee14,NingLiu14,Voorman14,Neykov15}
for high-dimensional models, some of which are reviewed by \cite{Dezeure14}.

\subsection{Group inference}

In this article, we consider a linear model
\begin{equation}\label{eq:model}
y=X \beta_0 + \varepsilon,
\end{equation}
where $\beta_0\in\R^p$ is the unknown parameter of interest, $y\in\R^n$ is a response vector, 
$X=[X_1\mid \cdots \mid X_p]\in \R^{n \times p}$ is a design matrix and
$\varepsilon\in\R^n$ is an i.i.d. error vector with mean zero and variance $\sigma^2$.
Define $\N_k=\{1,\ldots,k\}$ for an integer $k\geq 1$.
We are interested in making inference about a group of the parameters, $\beta_{0G}=(\beta_{0j})_{j\in G}$, for
$G\subset \N_p$ under a high-dimensional setting that $p>n$.
To be specific, the goal is to test the null hypothesis $H_{0G}:\beta_{0G}=0$ and construct
confidence regions for $\beta_{0G}$. These are arguably the most general inference problems,
obviously including individual inference about $\beta_{0j}$
as a special case when we choose $G$ to be a singleton.

Group inference arises naturally in applications where predictors have a block structure.
For instance, inference about a group of genomic loci within the same gene for its association with
a disease can identify responsive genes for the disease. 
Even if there is no application-driven block structure among the predictors,
group inference may still be useful. By grouping variables, one can detect signals that
are too small to detect individually. High correlation among predictors is 
a well-known difficulty for the Lasso and related individual inference approaches.
In this situation, grouping highly correlated predictors in the block Lasso will greatly stabilize 
the inference and increase detection power. Due to these advantages and practical usage,
a few methods have been proposed in recent papers for group inference.
A de-biased group Lasso is proposed by \cite{MitraZhang14} as a generalization of
the de-biased Lasso for more efficient group inference. \cite{vandeGeer15} define
a de-sparsified estimator for $\beta_{0G}$ with a surrogate Fisher information matrix constructed
by a multivariate square-root Lasso. 
\cite{Meinshausen15} develops the group-bound method to construct a one-sided confidence interval 
for $\|\beta_{0G}\|_{1}$ and shows that it is possible to detect the joint contribution of a group of highly
correlated predictors even when each has no significant individual effect. 
\cite{ZhouMain15} establish that a parametric bootstrap is asymptotically valid
for the group Lasso and demonstrate the advantages of grouping in finite-sample inference. 

A  large portion of the above methods perform statistical inference based on
the sampling distribution of an estimator $\estpivot=\estpivot(\hbeta,y,X)$ constructed as a function 
of $\hbeta$, which is either the Lasso or the group Lasso \citep{YuanLin06} depending
on whether group structure is used. 
Examples of such an estimator $\estpivot$ include the de-biased Lasso, 
the de-biased group Lasso, and
obviously the trivial case $\estpivot=\hbeta$ in those methods that directly estimate the distribution 
of $\hbeta$. 
There are two big challenges in these approaches. First, the finite-sample distribution of
$\estpivot$ is not well-understood, due to the high dimensionality $p>n$ and the sparsity in
$\hbeta$. Consequently, bootstrap has been used to approximate this distribution for inference.
Although the de-biased estimators have a nice asymptotic normal distribution when $n\to\infty$, 
they can be far from normally distributed when $n$ is finite. Indeed, some recent papers
propose to bootstrap the de-biased Lasso as a better alternative \citep{ZhangCheng16,Dezeure16}.
Then, here comes the second challenge: How to efficiently simulate from the bootstrap distribution
or an estimated sampling distribution of $\estpivot$? Without an explicit characterization of
the finite-sample distribution of the group Lasso (or Lasso), it appears that one can only rely
on bootstrap which can be computationally inefficient, or even impractical, for some
calculations, such as approximating tail probabilities in significance tests
and conditional sampling given a selected model in post-selection inference.

\subsection{Contributions of this work}

To meet the aforementioned challenges in group inference,
we develop the method of estimator augmentation for the block Lasso.
Partition the predictors into $J$ disjoint groups $\calG_j\subset \N_p$ for $j=1,\ldots,J$. 
For $\beta=(\beta_1,\ldots,\beta_p)$, let $\beta_{(j)}=(\beta_k)_{k\in\calG_j}$ for $j\in\N_J$.
Given $\alp\in[1,\infty]$,
the block Lasso is then defined via minimizing a penalized loss function $L(\beta;\alp)$:
\begin{equation}\label{eq:blocklassodef}
\hbeta \in \argmin_{\beta\in\R^p} \left\{ L(\beta;\alpha)\defi\frac{1}{2n} \|y-X \beta \|^2 + 
\lambda \sum_{j=1}^J w_j \| \beta_{(j)}\|_{\alpha}\right\},
\end{equation}
where $\|\cdot\|$ denotes the Euclidean norm. 
The weight $w_j >0$ usually depends on the group size $p_j = |\calG_j|$.
The regularizer is the block-$(1,\alpha)$ norm (when $w_j=1$) of $\beta$, hence the name block Lasso.
Note that the Lasso and the group Lasso can be regarded as the special cases of $\alp=1$ and $\alp=2$, respectively.
In the context of group inference, we can always choose a partition so that $G=\calG_j$ for some $j$,
which translates our task into inference about $\beta_{0(j)}$ using some function of $\hbeta$.
Instead of the distribution of $\hbeta$, we work on the joint distribution of a so-called 
augmented estimator $(\hbeta,S)$, where $S=S(y,X)\in\R^p$ is a vector. Under a particular choice of $S$,
we are able to obtain a closed-form density for the {\em exact} distribution of the augmented estimator
for any finite $n$ and $p$ and for all $\alp\in[1,\infty]$. Given the density, one may use Monte Carlo methods,
such as importance sampling, to draw from the joint distribution and simultaneously obtain samples of
$\hbeta$ and any function of $\hbeta$, such as the estimator $\estpivot$ used in an inferential method.
This method serves as a powerful and efficient alternative to parametric bootstrap
for $\estpivot$, and can potentially be applied in any group inference approach 
that utilizes some function of the block Lasso.
Estimator augmentation is especially useful in determining the significance in a hypothesis test
and approximating $[\estpivot\mid \hbeta \in B]$ for some event $B$.
In both scenarios, we need to sample on a rare event, which is known to be difficult
and sometimes impossible for bootstrap. We will demonstrate
such applications with two group inference approaches,
one using the group Lasso and the other a de-biased group Lasso.

Estimator augmentation was first developed for the Lasso in the work by \cite{Zhou14},
which does not respect any group structure.
Generalizing the method to the block Lasso for all block norms ($\alp\in(1,\infty]$) turns out 
to be very challenging technically. The sample space of the augmented estimator $(\hbeta,S)$,
which can be represented by a collection of manifolds with nonzero curvature, becomes more complicated. 
The joint distribution
is thus defined over a curved space, a significant distinction from the augmented Lasso estimator.
Along the development, we also identify a set of sufficient conditions
for solution uniqueness for the block Lasso, which are weaker and more transparent than known results.
When we were finalizing this manuscript, \cite{Tian16} posted a preprint in which they generalize
the technique of estimator augmentation to derive densities for selective sampling in a randomized
convex learning program. This exemplifies that estimator augmentation may have
much wider applications than what has been considered in our paper.

In addition to the above theoretical contributions, the significance of this work is also seen
from its application in group inference, especially when the group size $p_j$ is large.
\cite{MitraZhang14} prove that de-biasing a scaled group Lasso can achieve an efficiency gain
in group inference by a factor of $\sqrt{p_j}$ over a de-biased Lasso. 
\cite{ZhouMain15} show that, under certain conditions, bootstrapping the group Lasso
can reach a convergence rate of $n^{-1/2}$ if $\log J =O(p_j)$, which never holds for the Lasso ($p_j=1$)
in the high-dimensional setting $p\gg n$.
These results demonstrate the benefit of group sparsity in making inference about a group of parameters.
Our development of estimator augmentation for the block Lasso enables efficient simulation
from the sampling distributions of the group Lasso and the de-biased group Lasso, which is an essential
component in practical applications of these inferential approaches. 



Notation used in this paper is defined as follows.
Let $A\subset \N_p$ be an index set. 
For a vector $v=(v_j)_{1:p}$, we define $v_A=(v_j)_{j\in A}$.
For a matrix ${M}=(M_{ij})_{n\times p}$,
write its columns as $M_j$, $j=1,\ldots,p$, define ${M}_A=(M_{j})_{j\in A}$ as a matrix of size $n\times |A|$
consisting of columns in $A$, and similarly define ${M}_{BA}=(M_{ij})_{i\in B, j\in A}$ with $B\subset \N_n$.
Given the group structure $\calG$, let $\calG_A = \cup_{j\in A} \calG_j\subset \N_p$ for $A \subset \N_J$.
Define $v_{(A)} = v_{\calG_A}$
with the special case $v_{(j)}=v_{\calG_j}$ and let $G(v)=\{j\in\N_J: v_{(j)}\ne 0\}$ be the active groups of $v$.
For an $n \times p$ matrix ${M}$, ${M}_{(A)}={M}_{\calG_A}$, and for a $p \times p$
matrix ${M}$, ${M}_{(AB)}={M}_{\calG_A \calG_B}$, where $B \subset \N_J$. 
Denote by ${M}^+$ the Moore-Penrose pseudo-inverse of a matrix ${M}$ so ${M}^+=({M}^\trans {M})^+ {M}^\trans$
when ${M}$ is not a square matrix. 
We use $\row(M)$ and $\nul(M)$ to denote the row space and the null space of $M$, respectively.
Let $\diag({M},{M}')$ be the block diagonal matrix with ${M}$ and ${M}'$ as the diagonal blocks.
Denote by $\phi_n(\bullet;c)$ the density of $\dnorm_n(0,c\bfI_n)$ for $c>0$. 
Let $\mbS^{m-1}_{\alp}=\{v\in\R^m: \|v\|_{\alp}=1\}$ be 
the unit $\ell_{\alp}$-sphere in $\R^m$.
We may suppress $(m-1)$ and simply write $\mbS_{\alp}$ when the dimension 
does not need to be specified explicitly.

Throughout the paper, let $\alpha^*$ be conjugate to $\alpha$ in the sense that $\frac{1}{\alpha}+\frac{1}{\alp^*}=1$.
We will assume that $\alpha\in(1,\infty)$ unless noted otherwise in next three sections, 
and leave to Section~\ref{sec:inftynorm}
the special case $\alpha=\infty$ whose technical details are slightly more complicated.
Although not the focus of this paper, the results for the Lasso can be obtained 
as another special case ($\alp=1$) after some simple modifications 
of the corresponding results for $\alp\in(1,\infty)$.

\section{The Basic Idea}\label{sec:basic}

In this section, we give an overview of the idea of estimator augmentation. 
We start with 
the Karush-Kuhn-Tucker (KKT) conditions for the minimization problem~\eqref{eq:blocklassodef}.
Under uniqueness of the block Lasso, we will establish a bijection between $y$ and the augmented
estimator and derive the joint density of its sampling distribution. We note that 
solution uniqueness for the block Lasso is an interesting topic in its own right, and the sufficient conditions
in this work are much more transparent than those in the existing literature. 

\subsection{The KKT conditions}

Denote by $\sgn(\cdot)$ the sign function with the convention that $\sgn(0)=0$.
For a scalar function $f: \R\to\R$ and a vector $v=(v_i)\in\R^m$, 
we define 
\begin{equation}\label{eq:fonvec}
f(v):=(f(v_1),\ldots,f(v_m)).
\end{equation}

\begin{definition}\label{def:rhofunction}
For $\alp\in(1,\infty)$, let $\rho=\alp^*/\alp \in (0,\infty)$ and define $\eta:[-1,1]\to[-1,1]$ by
\begin{equation*}
\eta(x)=\eta(x;\rho)=\sgn(x)|x|^\rho.
\end{equation*}
Denote its inverse function by $\eta^{-1}(x)=\sgn(x)|x|^{1/\rho}$.
\end{definition}

\noindent Some basic properties about $\eta$ are given in Lemma~\ref{lm:rhoStoS} in Appendix~\ref{sec:app_uniqueness}. In particular, $\eta(v)$ for $v\in\mbS_{\alp^*}$, 
interpreted in the sense of \eqref{eq:fonvec}, 
is a bijection from $\mbS_{\alp^*}$ onto $\mbS_{\alp}$. This fact is used in \eqref{eq:Sdef} below.

Let $S=(S_1,\ldots,S_p) \in \R^p$ such that $S_{(j)} \in \R^{p_j}$ is a subgradient of $\|\beta_{(j)}\|_{\alp}$ 
evaluated at the solution $\hbeta_{(j)}$ of \eqref{eq:blocklassodef}.
According to Lemma~\ref{lm:subdiff} on the subdifferential of $\|\cdot\|_{\alp}$, we have
\begin{align}\label{eq:Sdef}
\begin{cases}
S_{(j)} = \eta^{-1}(\hbeta_{(j)}/\|\hbeta_{(j)}\|_{\alp})\in\mbS^{p_j-1}_{\alp^*} &\text{ if } \hbeta_{(j)} \ne 0  \\
\| S_{(j)}\|_{\alp^*} \leq 1 &\text{ if } \hbeta_{(j)}=0. 
\end{cases}
\end{align}
For the special case $\alp=\alp^*=2$ (group Lasso), $\eta(v)=\eta^{-1}(v)=v$ and
the above subgradient reduces to the familiar result in \cite{YuanLin06}.
Put ${W}=\diag(w_1 \bfI_{p_1},\ldots,w_J \bfI_{p_J})\in\R^{p\times p}$.
The KKT conditions of \eqref{eq:blocklassodef}, which are both sufficient and necessary, are
\begin{equation}\label{eq:KKTwhole}
\frn X^{\trans} X \hbeta + \lambda {W} S = \frn X^{\trans} y
\end{equation}
for a vector $S$ satisfying \eqref{eq:Sdef}.

\begin{definition}\label{def:ae}
Let $S$ be defined by \eqref{eq:Sdef} and \eqref{eq:KKTwhole}.
We will call $(\hbeta,S)\in\R^{2p}$ an augmented solution to the block Lasso problem \eqref{eq:blocklassodef}.
When we study the sampling distribution of $\hbeta$, the random vector $(\hbeta,S)$ will be
called an {\em augmented estimator}. 
\end{definition}

If $(\hbeta,S)$ is unique for each $y$,
then \eqref{eq:KKTwhole} defines a bijective mapping from the space of $(\hbeta,S)$ onto 
the space of $y$, 
which is the inverse of the minimization program \eqref{eq:blocklassodef} that maps $y$ to $(\hbeta,S)$.
From the density of $y$ or $\veps$, it is hopeful to derive the
joint density of the augmented estimator $(\hbeta,S)$ via this bijective mapping. 
Then one may apply Monte Carlo methods, such as Markov chain Monte Carlo (MCMC) and importance sampling,
to draw from the joint distribution of the augmented estimator. 
As a marginal distribution, the sampling distribution of $\hbeta$ can be
readily approximated by Monte Carlo samples, as well as any function of $(\hbeta,S)$. 
This is the basic idea of estimator augmentation.
Although the idea seems intuitive, there are a few technical difficulties in the implementation:
\begin{enumerate}
\item To establish the uniqueness of $(\hbeta,S)$ under fairly general situations.
\item To characterize the sample space for $(\hbeta,S)$, which appears to be a $2p$-vector
but in fact lives in the union of a finite number of $n$-dimensional manifolds. This makes the aforementioned
bijection conceivable since $\veps\in\R^n$.
\item To calculate the Jacobian of the mapping and obtain the target density via a change of variable.
\end{enumerate}
We will establish the solution uniqueness in the remainder of this section,
and take care of the other two major steps in Section~\ref{sec:EA}.
Although the basic idea follows from that in \cite{Zhou14}, there are substantial new technical issues
in each of the three steps, which will be discussed in the sequel.

\subsection{Uniqueness}

We briefly present here
the most relevant results about solution uniqueness for the block Lasso, 
while leaving many useful intermediate results 
and proofs to Appendix~\ref{sec:app_uniqueness}.
Despite that the KKT conditions only require the existence of a subgradient, it turns out that $S$
is always unique:

\begin{lemma}\label{lm:fittedandS}
For any $y$, $X$ and $\lambda>0$, every $\hbeta$ \eqref{eq:blocklassodef} for $\alp\in[1,\infty]$ gives the same fitted value $X\hbeta$ and the same subgradient $S$.
\end{lemma}

This lemma covers the full domain of $\alp$, including the boundary cases $\alp=1$ (Lasso) and $\alp=\infty$.
Hereafter, we call $S$ ``the" subgradient vector due to its uniqueness.
Let $U = \frn X^{\trans} \veps \in \R^p$, and denote the Gram matrix by $\Psi=\frn {X}^{\trans} {X}$ hereafter. 
The KKT conditions \eqref{eq:KKTwhole} can be written as
\begin{equation}\label{eq:KKTU}
\Psi \hbeta + \lambda {W} S - \Psi \beta_0=U.
\end{equation}
Since $U \in \row(X)$ and $\Psi(\hbeta-\beta_0) \in \row(X)$, we have
\begin{equation}\label{eq:SinrowX}
{W} S \in \row(X) \Leftrightarrow S\in \row(X{W}^{-1}):=\calV \subset \R^p.
\end{equation}
Next, we state sufficient conditions for the uniqueness of $\hbeta$ and thus the uniqueness of 
the augmented solution $(\hbeta,S)$.

\begin{assumption}\label{as:X}
Every $(n \wedge p)$ columns of $X$ are linearly independent.
\end{assumption}

\begin{definition}\label{def:genpos}
We say that the columns of a matrix $M\in\R^{n\times p}$ is in blockwise general position with respect
to $(\calG,\alp)$ if for all $s\in\row(M)$, the vectors 
$\{M_{(j)}\eta({s}_{(j)}): j\in \calE\}$ are in general position,
where $\calE=\{j\in\N_J:\|s_{(j)}\|_{\alp^*}=1\}$.
\end{definition}

\begin{assumption}\label{as:pj}
(a) The block size $p_j=|\calG_j|\leq n$ for all $j\in\N_J$.
(b) The columns of $XW^{-1}$ are in blockwise general position with respect to $(\calG,\alp)$.
\end{assumption}

The two assumptions are quite weak. 
Assumption~\ref{as:X} simply states that $X$ does not satisfy
any additional linear constraint other than those that must be satisfied by any $n\times p$ matrix.
If the entries of $X$ are drawn from a continuous distribution,
then Assumption~\ref{as:X} holds with probability one. 
Assumption~\ref{as:pj}(a) can be regarded as a minimum sample size requirement, $n \geq m=\max_j p_j$.
To establish estimation consistency for the block Lasso, one needs the scaling that
$(m\vee \log J)/n \to 0$ \citep{Negahban12,ZhouMain15}, which is stronger than this assumption. 
To help understand the intuition behind Assumption~\ref{as:pj}(b), choose $W=\bfI_p$ to simplify the exposition.
Then this assumption is imposed on the vectors $X_{(j)}v_{(j)}$, where 
$v_{(j)}=\eta(s_{(j)})\in \mbS_{\alp}$ (Lemma~\ref{lm:rhoStoS}) and $s\in\calV$.
Under Assumption~\ref{as:X} with $n\leq p$, 
$\dim(\calV)=n$ and $v=\eta(s)\in\R^p$ has only $n$ free coordinates.
Thus, we essentially require  linear combinations of disjoint subsets of any $n$ columns of $X$ be in general
position, which is a mild condition in practice. Note that $p_j\leq n$ ensures that the vector $X_{(j)}v_{(j)}$
can be uniquely represented in $\spn(X_{(j)})$.
For the special case of Lasso,
Assumption~\ref{as:pj}(b) reduces to that the columns of $X$ are in general position, the same as in \cite{Tibshirani13}. 

\begin{theorem}\label{thm:uniqsuff}
Suppose that Assumptions~\ref{as:X} and \ref{as:pj} hold. 
For any $\lambda>0$ and $y\in\R^n$,
the solution $\hbeta$ to the block Lasso problem \eqref{eq:blocklassodef} with $\alp\in[1,\infty)$ is unique 
and $|G(\hbeta)|\leq n\wedge J$.
\end{theorem}

Since solution uniqueness is a topic of independent interest, we make a brief comparison
to some existing results. Theorem~\ref{thm:uniqsuff} 
unifies a few important special cases, including the Lasso ($\alp=1$) and the group Lasso ($\alp=2$).
For $\alp=1$, this theorem is comparable to the result in \cite{Tibshirani13},
while the existing results about the uniqueness of the group Lasso involve conditions that are much less
transparent than the ones stated here. 
As an example, Theorem 3 in \cite{RothFischer08} states that, under Assumption~\ref{as:X}, 
the group Lasso solution $\hbeta$ (with $\alp=2$) is unique
if (i) $|\calG_A|\leq n$, where $A=G(\hbeta)$ is the
active groups, and (ii) $A=\{j\in\N_J: \|S_{(j)}\|=1\}$. Unlike Assumption~\ref{as:pj} which is
imposed on $X$ explicitly, conditions (i) and (ii) are implicit in nature and 
can be verified only after a particular solution is calculated.
According to Theorem~\ref{thm:uniqsuff}, it is possible to have a unique solution 
when $|\calG_A|>n$ as long as $|A|\leq n$, i.e., there are no more than $n$ active groups but
the total number of active coefficients of $\hbeta$ could be greater than the sample size.
Such cases are not covered by the result in \cite{RothFischer08}.
As will become clear in next section, the set of $\hbeta$ satisfying (i) and (ii) is
a proper subset of the full space of unique solutions and thus, in general, will have a probability mass 
strictly less than one.

\section{Estimator Augmentation}\label{sec:EA}

We will go though the main steps in detail to derive the joint density of the augmented estimator $(\hbeta,S)$,
which is useful for understanding this method. Section~\ref{sec:space} characterizes the sample space
of $(\hbeta,S)$, Section~\ref{sec:mapping} defines explicitly the bijective mapping from the KKT conditions,
and Section~\ref{sec:density} derives the joint density.
A few concrete examples will follow in Section~\ref{sec:examples} to illustrate the method.
The joint density of $(\hbeta,S)$ depends on the true parameter $\beta_0$ and the error distribution.
We will discuss in Section~\ref{sec:infer} how to
apply estimator augmentation in high-dimensional inference.
By default, we assume $p\geq J\geq n$. The results for $p<n$ will be obtained as special cases.

\subsection{Sample space}\label{sec:space}

Denote by $\hgamma=(\hgamma_j)\in\R^J$ the vector of the block norms of $\hbeta$, i.e.
$\hgamma_j=\|\hbeta_{(j)}\|_{\alp}$. 
It follows from \eqref{eq:Sdef} that $\hbeta_{(j)}=\hgamma_j \eta(S_{(j)})$ for all $j\in\N_J$.
Thus, the augmented estimator $(\hbeta,S)$ can be represented by the triplet $(\hgamma_{\calA},S,\calA)$,
where $\calA=G(\hbeta)$ is a random subset of $\N_J$ when considering the sampling distribution.
Given $\calA=A$ for a fixed subset $A\subset \N_J$, it is seen from \eqref{eq:Sdef} and \eqref{eq:SinrowX} that the sample space for $S$ is 
\begin{equation}\label{eq:defcalSofA}
\calM_A=\left\{s\in  \calV: \|s_{(j)}\|_{\alp^*}=1\;\forall j\in A \text{ and } \|s_{(j)}\|_{\alp^*}\leq 1\;\forall j\notin A\right\}.
\end{equation}
Since $s_{(j)}\in \mbS^{p_j-1}_{\alp^*}$ for $j\in A$ and $\dim(\calV)=n$ under Assumption~\ref{as:X}, 
$\calM_A$ is an $(n-|A|)$-manifold in $\R^{p}$ if $|A|\leq n$
and it is the product of unit $\ell_{\alp^*}$-spheres and balls intersecting with the linear subspace $\calV$.
Correspondingly, the space for $(\hgamma_{A},S)$ given $A$ is
$\Omega_A=(\R^+)^{|A|}\times \calM_A$, which is an $n$-manifold.
Taking union over subsets of size $\leq n$, we obtain the sample space for the augmented estimator
$(\hgamma_{\calA},S,\calA)$:
\begin{equation}\label{eq:defOmg}
\Omega=\bigcup_{|A|\leq n} \Omega_A \times \{A\}.
\end{equation}

\begin{remark}
We do not have to consider $\{|\calA|>n\}$, since this never happens
under the assumptions of Theorem~\ref{thm:uniqsuff}. Hereafter, we always regard the essential range of $\calA$ as
\begin{equation}\label{eq:Arange}
\scrA:=\{A\subset \N_J: |A|\leq n\}.
\end{equation}
\end{remark}

In summary, the sample space of the augmented estimator, represented by the triplet $(\hgamma_{\calA},S,\calA)$,
is the union of a finite number of $n$-manifolds. Thus, it is possible
to find a bijective mapping from this space to $\R^n$, the space for $\veps$.
For the Lasso, $\Omg_A$ degenerates to an $n$-dimensional polyhedron with zero curvature.

\begin{remark}\label{rmk:sepnorm}
Parameterizing the augmented estimator in terms of $\hgamma$ and $S$ is a critical choice for our derivations.
In this way, all the equality constraints are imposed on $S$ as in \eqref{eq:defcalSofA},
leading to familiar geometry for the spaces of $\hgamma$ and $S$, which is helpful
for understanding distributions over these spaces. It is also a nature choice, since 
the subgradient $S$ is alway unique (Lemma~\ref{lm:fittedandS}) 
and non-uniqueness comes solely from $\hgamma$ (Lemma~\ref{lm:conditionuniq}).

\end{remark}

\subsection{A bijective mapping}\label{sec:mapping}

Putting $\hbeta_{(j)}=\hgamma_j\eta(S_{(j)})$ for $j\in\calA$,
Equation \eqref{eq:KKTU} becomes
\begin{equation}\label{eq:defineH}
\frn X^{\trans}\veps=\sum_{j\in\calA}\hgamma_{j}\Psi_{(j)} \eta(S_{(j)}) + \lambda {W} S - \Psi \beta_0
:=H(\hgamma_{\calA},S,\calA;\beta_0,\lambda),
\end{equation}
which defines a mapping $H:\Omega \to \row(X)$ for any $\beta_0 \in \R^p$ and $\lambda>0$.
For notational brevity, we often suppress its dependence on $(\beta_0,\lambda)$ and write the mapping as $H(\bullet)$.
In what follows, we show that $H$ is bijective, which is a consequence of
the uniqueness of $(\hbeta,S)$, or equivalently of $(\hgamma_{\calA},S,\calA)$.

\begin{lemma}\label{lm:bijection}
Suppose $\alp\in[1,\infty)$ and that Assumptions~\ref{as:X} and \ref{as:pj} hold. 
Then  for any $\beta_0\in\R^p$ and $\lambda>0$, 
$H$ is a bijection that maps $\Omg$ onto $\row(X)$.
\end{lemma}

\noindent
This lemma applies to $\alp=1$, in which case we define 
$\eta(x)=xI(|x|=1)$ and $\eta^{-1}(x)=\sgn(x)$ 
by taking the limit $\rho\to \infty$ in Definition~\ref{def:rhofunction}.

The mapping $H$ is established at a quite abstract level so far. 
It will be more convenient to work with the restriction of $H$ to $\Omg_A$ for $A\in\scrA$, defined by
\begin{equation}\label{eq:defineHA}
H_A(r_A,s):=H(r_A,s,A) \quad\text{ for } (r_A,s) \in \Omg_A,
\end{equation}
where $r=(r_1,\ldots,r_J)\in\R^J$ with $r_j=0$ for $j\notin A$. Then $H$ can be understood as
a collection of bijective mappings $\{H_A: A\in\scrA\}$ indexed by subsets of $\N_J$.
Write the block Lasso solution for the response $y$ as $\hbeta=\hbeta(y)$. Let 
\begin{align*}
E_A:=\left\{v\in\R^n: G\left(\hbeta(X\beta_0+v)\right)=A \right\}
\end{align*}
be the set of noise vectors $v$ for which the active set of the block Lasso solution 
$\hbeta(X\beta_0+v)$ is $A$. Denote the block norms and the subgradient of $\hbeta(X\beta_0+v)$
by $\hgamma(X\beta_0+v)$ and $S(X\beta_0+v)$, respectively. 
Then for $v\in E_A$, we have 
\begin{align*}
H_A(\hgamma_A(X\beta_0+v),S(X\beta_0+v))=\frn X^\trans v.
\end{align*}
Now the bijective nature of $H_A$ allows us to obtain the density for $(\hgamma_A,S)$ from
the density of the noise vector via a change of variable.

It remains to find the differential of $H_A$ so that we can calculate the Jacobian for the change of variable.
A special aspect of this mapping is that $H_A$ is defined on a manifold and thus
its differential is determined with respect to local parameterizations. 
As an $(n-|A|)$-manifold in $\R^p$, a neighborhood of $s\in\calM_A$ can be parameterized 
by $s_F$, where $F\subset \N_p$ may depend on $(s,A)$ and $|F|=n-|A|$.
Correspondingly, the $n$-manifold $\Omg_A$ will be parameterized by $(r_A,s_F)\in\R^n$
in a neighborhood of $(r_A,s)$. Under this parameterization, Lemma~\ref{lm:diffHA}, 
proven in Appendix~\ref{sec:pfdHA}, gives 
an expression of $d H_A$ in terms of a few matrices defined below.
Let $\eta'(x)=\rho |x|^{\rho-1}$ denotes the derivative of $\eta$. Define
\begin{align}
r\circ \Psi & :=[r_1\Psi_{(1)}|\ldots | r_J\Psi_{(J)}] \in \R^{p\times p} \label{eq:defrPsi}\\
\Psi\circ \eta & :=[\Psi_{(1)}\eta(s_{(1)})|\ldots | \Psi_{(J)}\eta(s_{(J)})] \in \R^{p\times J} \label{eq:defPsis}
\end{align}
and $D=D(s,A)\in\R^{p\times p}$ to be a diagonal matrix whose diagonal elements
$D_{kk}=\eta'(s_k)$ for $k\in \calG_A$ and $D_{kk}=0$ otherwise.

\begin{lemma}\label{lm:diffHA}
Fix $p\geq n$, $\beta_0\in\R^p$, $\lambda>0$ and $A\in \scrA$.
Suppose that $\alp\in(1,\infty)$ and Assumption~\ref{as:X} holds. Then for any interior point $(r_A,s)\in\Omg_A$,
there is a full rank matrix $T=T(\eta(s),A)$ of size $p\times (n-|A|)$ such that $ds=T(\eta(s),A) ds_F$ and
\begin{equation}\label{eq:defM}
d H_A  = \left[(\Psi\circ \eta)_A \mid \{(r\circ \Psi)D + \lambda W\}T(\eta(s),A)\right] d\theta 
:= M(r_A,s,A;\lambda) d\theta,
\end{equation}
where $\theta=(r_A,s_F)\in\R^n$. 
\end{lemma}

\begin{remark}
This lemma applies to every interior point of $\Omg_A$, irrespective of whether or not the corresponding
solution is unique.
Assumption~\ref{as:X} is only needed to fix the dimension of the manifold $\calM_A$.
With some modifications of the proof, the result can be generalized to the situation 
in which Assumption~\ref{as:X} fails to hold.
The size of the matrix $M=M(r_A,s,A;\lambda)$ is $p\times n$, and if Assumption~\ref{as:pj} also holds 
it will be full rank.
The parameterization $s_F$ for $\calM_A$ is defined locally for a neighborhood of $s$.
For each $j\in A$, the unit sphere $\mbS^{p_j-1}_{\alp^*}$, except a set of measure zero,
can be covered by two parameterizations, one for each open semi-sphere.
\end{remark}

\begin{remark}\label{rm:MforgrpLasso}
For the special case $\alp=\alp^*=2$ (group Lasso), we have 
$\rho=1$, $\eta(x)=x$ and $\eta'(x)=1$ for $x\in\R$. The matrix $M$ \eqref{eq:defM} has a simpler form:
\begin{equation}\label{eq:Mfor2}
M(r_A,s,A;\lambda)=\left[(\Psi\circ s)_A \mid (r\circ \Psi + \lambda W)T(s,A)\right]. 
\end{equation}
The only reason that we excluded the case $\alp=1$ in the above lemma is because $\eta'$ is not well-defined.
We will cover this case, which reduces to the Lasso, in Example~\ref{exp:lasso}.
\end{remark}

Geometrically, the columns of $T$ consist of a set of tangent vectors of the manifold $\calM_A$
while those of $M$ consist of tangent vectors of the mapping $H_A$. These tangent vectors
determine the ratio between the volume element in the image space $\row(X)$ and that in the domain $\Omg_A$, 
and thus the Jacobian of the mapping. 

\subsection{Joint density}\label{sec:density}

Now we make an explicit link from the augmented estimator $(\hgamma_{\calA},S,\calA)$ to the noise vector $\veps$.
Under Assumption~\ref{as:X}, $\nul(X^\trans)=\{0\}$ and thus by \eqref{eq:defineH}
\begin{equation}\label{eq:bieps}
\veps/\sqn=\sqn (X^\trans)^+ H(\hgamma_{\calA},S,\calA;\beta_0,\lambda)
:=\tdH(\hgamma_{\calA},S,\calA;\beta_0,\lambda).
\end{equation}
We note that $\tdH\in\R^n$ is the coordinates of $H$ with respect to the basis 
$(X^{\trans}/\sqn)$ of $\row(X)$.
By Lemma~\ref{lm:bijection}, $\tdH$ is a bijection that maps $\Omg$ onto $\R^n$.
Define $\tdH_A$ similarly as for $H_A$ in \eqref{eq:defineHA}.
It then follows from Lemma~\ref{lm:diffHA}  that the Jacobian of $\tdH_A$ is 
\begin{equation}\label{eq:jacob}
J_A(r_A,s;\lambda)=\det \left[\sqn (X^\trans)^+ M(r_A,s,A;\lambda)\right],
\end{equation}
of which the matrix on the right side is of size $n\times n$.

\begin{theorem}\label{thm:densityhigh}
Fix $p\geq n$, $\beta_0\in\R^p$ and $\lambda>0$. Suppose that Assumptions~\ref{as:X} and \ref{as:pj} hold. 
Then the distribution of the augmented estimator $(\hgamma_{\calA},S,\calA)$ for $\alp\in(1,\infty)$ is given by
the differential form
\begin{equation}\label{eq:nform}
d\mu_A:=\Prob(d r_{A}, ds,\{A\}) = g_n(\tdH_A(r_A,s;\beta_0,\lambda))|J_A(r_A,s;\lambda)| d\theta
\quad\text{ for }(r_A,s,A)\in \Omg,
\end{equation}
where $\theta=(r_A,s_F)\in\R^n$ and $g_n$ is the density of $(\veps/\sqn)$.
\end{theorem}

See Appendix~\ref{sec:proofhimden} for a proof, 
from which we see that \eqref{eq:nform} is valid as long as
the block Lasso program \eqref{eq:blocklassodef} has a unique solution for almost all $y\in\R^n$.
For each $A\in\scrA$, the $n$-form $d\mu_A$ defines a measure on $\Omg_A$ in the
following sense. Let $k=n-|A|$ and
\begin{equation}\label{eq:density}
f_A(r_A,s)=g_n(\tdH_A(r_A,s))|J_A(r_A,s)|.
\end{equation}
Suppose that $\Gamma\subset (\R^+)^{|A|}$ and 
$\Phi=\{\Phi(u):u\in\Delta\} \subset \calM_A$ is a $k$-surface in $\R^p$
with parameter domain $\Delta\subset \R^k$.
Then by \eqref{eq:nform} we have
\begin{equation}\label{eq:probability}
\Prob(\hgamma_A \in \Gamma, S\in \Phi, \calA = A)= \int_{\Gamma \times \Phi} d\mu_A
=\int_{\Gamma} \int_\Delta f_A(r_A,\Phi(u)) \left|\frac{\partial s_F}{\partial u}\right| du \,dr_A,
\end{equation}
where the Jacobian ${\partial s_F}/{\partial u}=1$ if $\Phi$ is parameterized by $s_F$.
Note that for a particular $k$-surface, parameterizations other than by $s_F$ may be more convenient.
As shown in \eqref{eq:probability}, 
the distribution of $(\hgamma_{\calA},S,\calA)$ is defined by a collection
of measures, $\{\mu_A: A\in\scrA\}$, due to the discrete nature of $\calA$,
and $f_A$ is the density of $\mu_A$ parameterized by $\theta=(r_A,s_F)$.
An important special case of the above integral is
\begin{equation*}
\Prob(\calA=A)=\mu_A(\Omg_A)=\int_{\Omg_A} d\mu_A.
\end{equation*}
Lastly, summing over $\scrA$ in the above equation leads to
\begin{equation*}
\sum_{A\in\scrA} \int_{\Omg_A} d\mu_A=\sum_{|A|\leq n} \Prob(\calA=A) =1.
\end{equation*}

\begin{remark}\label{rmk:ondensity}
Evaluation of the joint density in \eqref{eq:nform} for any $(r_A,s,A)\in\Omg$ can
be done by a simple procedure: 
\begin{enumerate}
\item Find a local parameterization $s_F$ for $s$ and the associated matrix $T$;
\item Calculate the Jacobian $J_A$ \eqref{eq:jacob}, evaluate the mapping $\tdH_A(r_A,s)$ \eqref{eq:bieps}, 
and plug them into \eqref{eq:density} to obtain $f_A(r_A,s)$.
\end{enumerate}
See Examples~\ref{exp:simple} and \ref{exp:orthogonal} for concrete illustrations.
\end{remark}

As a consequence of Theorem~\ref{thm:densityhigh}, the density $f_A$ under i.i.d. Gaussian errors can be found in the
following corollary. See Appendix~\ref{sec:pfGaussian} for a proof.
\begin{corollary}\label{cor:dGaussian}
Suppose that the assumptions of Theorem~\ref{thm:densityhigh} hold. If $\veps\sim \dnorm_n(0,\sigma^2\bfI_n)$,
then for $(r_A,s)\in\Omg_A$ and $A\in\scrA$,
\begin{align}\label{eq:dGaussian}
f_A(r_A,s)=\left(\frac{2\pi\sigma^2}{n}\right)^{-n/2}
\exp\left[-\frac{1}{2\sigma^2}\|X(b+\lmd \Psi^+Ws-\beta_0)\|^2\right]|J_A(r_A,s)|,
\end{align}
where $b\in \R^p$ is given by $b_{(j)}=r_j \eta(s_{(j)})$ for all $j\in\N_J$.
\end{corollary}

As we have seen, the sample space $\Omg$ \eqref{eq:defOmg} for the augmented estimator is complex
due to the many constraints involved in $\calM_A$ \eqref{eq:defcalSofA} and the mix of continuous
and discrete components. It is quite surprising that one can
find an exact joint density for the augmented estimator given $\beta_0$ and the noise distribution which 
is usually simple under an i.i.d. assumption. The density gives a complete and explicit characterization of
the sampling distribution according to \eqref{eq:probability}. 
In light of the non-linear and sparse nature of $\hbeta$ and the high-dimension of
the problem, the joint density itself is a significant theoretical result that improves our understanding
of the block Lasso estimator. Applications of this result 
in group inference will be discussed in Section~\ref{sec:infer}.

\begin{remark}\label{rmk:differences}
We summarize the main differences between the joint density of the augmented block Lasso 
in Theorem~\ref{thm:densityhigh} and that of the augmented Lasso in \cite{Zhou14}.
First, the sample space $\Omg_A$ is an $n$-manifold with nonzero curvature for $\alp>1$, 
and consequently the density is specified in \eqref{eq:nform} by a differential form of order $n$.
In contrast, the space $\Omg_A$ has no curvature for the augmented Lasso estimator, whose density
can be defined with respect to the Lebesgue measure. Second, the Jacobian \eqref{eq:jacob} depends
on both $s$ and $A$ for the block Lasso, while it only depends on $A$ for the Lasso. 
See Example~\ref{exp:lasso} for the technical reason and a geometric interpretation for these differences.
Both aspects result in
new challenging computational issues in this work for the development of Monte Carlo algorithms,
which are discussed in Section~\ref{sec:is}.
\end{remark}

For the sake of completeness, we also give the density of $(\hgamma_{\calA},S,\calA)$ when $p<n$,
which can be obtained by simple modifications of a few steps in the proof of the result for $p\geq n$.
See Appendix~\ref{sec:pflowden} for detail.
Assume that $\rank(X)=p<n$, which is sufficient for
both Assumptions~\ref{as:X} and \ref{as:pj} to hold.
Then $\row(X)$ and $\calV$ \eqref{eq:SinrowX} are identical to $\R^p$, 
which implies every $s\in\calM_A$ \eqref{eq:defcalSofA} can be locally parameterized by $s_F$ with
$|F|=p-|A|$. In this case, $H_A:\Omg_A \to \R^p$ and $M(r_A,s,A;\lambda)\in\R^{p\times p}$.

\begin{corollary}\label{cor:densitylow}
Fix $n> p$, $\beta_0\in\R^p$ and $\lambda>0$. If $\rank(X)=p$,
the distribution of the augmented estimator $(\hgamma_{\calA},S,\calA)$ for $\alp\in(1,\infty)$ is given by
the $p$-form
\begin{equation}\label{eq:pform}
d\mu_A= g_{n,X}(H_A(r_A,s;\beta_0,\lambda))|\det M(r_A,s,A;\lambda)| d\theta
\quad\text{ for }(r_A,s,A)\in \Omg,
\end{equation}
where $\theta=(r_A,s_F)\in\R^p$ and $g_{n,X}$ is the density of $U=n^{-1} X^\trans \veps$.
\end{corollary}

An interesting observation is that we need the density $g_{n,X}$ of $U=n^{-1} X^\trans \veps$
when $p<n$, which is more difficult to determine than the density of $\veps$ needed in the high-dimensional
case \eqref{eq:nform}. The underlying reason for this can be found from the sufficient statistic $t=X^\trans y$ \eqref{eq:KKTwhole}. When $p<n$, the dimension of $t$ is smaller than the sample size $n$
and thus this statistic achieves the goal of data reduction. Consequently, one needs the distribution of $t$ or $U$
to determine the sampling distribution of $\hbeta$. However, when $p\geq n$, $y$ is the coordinates of
the statistic $t$ using the rows of $X$ as the basis and thus the two are equivalent up to a change of basis,
in which case the distribution of $y$ or $\veps$ is all we need. 

\subsection{Examples}\label{sec:examples}

We illustrate the distribution of the augmented estimator with a few examples. Example~\ref{exp:simple}
is a simple concrete example that demonstrates various key concepts, 
including the sample space, the density, and probability calculations. 
The second example shows that, under an orthonormal design, the joint distribution   
given by Theorem~\ref{thm:densityhigh} coincides with the result from block soft-thresholding.
The last example considers the Lasso.
Technical details involved in these examples are deferred to Appendix~\ref{sec:techforexps}.

\begin{example}\label{exp:simple}
Consider a simple but nontrivial example with $p=3$, $n=2$, and $J=2$. The two groups
$\calG_1=\{1,2\}$ and $\calG_2=\{3\}$, and pick $\alp=2$. Suppose that
\begin{align*}
\frac{1}{\sqn}X=\left[\begin{array}{ccc}
1 & 0 & 1 \\
0 & 1 & 1
\end{array}\right],
\quad\quad \beta_0=0, \quad\quad \veps\sim \dnorm_2(0,\sigma^2\bfI_2),\quad\quad W=\bfI_3.
\end{align*}
Put $r=(r_1,r_2)$ and $s=(s_1,s_2,s_3)$.

We first determine the space $\calM_A$ \eqref{eq:defcalSofA}. Incorporating the constraint that
\begin{equation}\label{eq:exp1rowspace}
s\in\calV=\row(X) \Leftrightarrow s_1+s_2-s_3=0,
\end{equation}
the manifold $\calM_A$ can be expressed as
\begin{equation}\label{eq:exp1MA}
\calM_A=\{(s_1,s_2,s_1+s_2): (s_1,s_2)\in \mbD_A\},
\end{equation}
where $\mbD_A\subset \R^2$ is the range for $s_{(1)}=(s_1,s_2)$. 
Let $\mbB^m$ be the unit $\ell_2$-ball in $\R^m$. 
For $A=\varnothing$, the definition of $\calM_A$ shows that
$\mbD_{\varnothing}=\mbB^2 \cap \{|s_1+s_2|\leq 1\}$,
whose boundary $\partial \mbD_{\varnothing}$ consists of two arcs and two line segments
connecting at four points:
$a=(1,0)$, $b=(0,1)$, $c=(-1,0)$, and $d=(0,-1)$.
See Figure~\ref{fig:space} for illustration.
Use $\partial(q_1, q_2)$ to denote the boundary of $\mbD_{\varnothing}$
from $q_1$ to $q_2$ along the positive orientation. 
It is then immediate that 
\begin{equation}\label{eq:exp1DA}
\mbD_{A}=\left\{
\begin{array}{ll}
\partial(b,c) \cup \partial(d,a) & \text{ for }A=\{1\} \\
\partial(a,b) \cup \partial(c,d) & \text{ for }A=\{2\} \\
\{a,b,c,d\} & \text{ for }A=\{1,2\}.
\end{array}
\right.
\end{equation}
Plugging $\mbD_{A}$ back to \eqref{eq:exp1MA}, we see that $\calM_\vn$ is a surface,
$\calM_{\{1\}}$ and $\calM_{\{2\}}$ are curves, and
$\calM_{\{1,2\}}$ degenerates to four points in $\R^3$. 

\begin{figure}[!ht]
  \centering
 \includegraphics[width=0.35\textwidth,angle=-90,trim=0.5in 1in 1in 0.5in,clip]{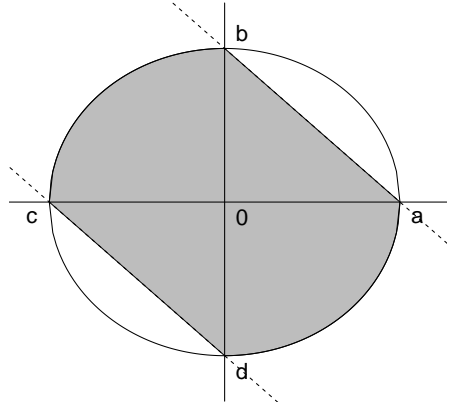}
  \caption{Sample space of $S_{(1)}$ shown as the shaded area.\label{fig:space}}
\end{figure}

We find $f_A$ \eqref{eq:density} and calculate $\Prob(\calA=A)$ for $A=\{1\}$ here.
The two arcs in $\mbD_{\{1\}}$ can be parameterized by $s_1$ and
$\Omg_{\{1\}}$ correspondingly by $\theta=(r_1,s_1)$ with two domains,
$\R^+\times (-1,0)$ and $\R^+\times (0,1)$.
After a few steps of algebra, we arrive at the density
\begin{align}
f_{\{1\}}(r_1,s_1,s_2)=\frac{1}{\pi \sigma^2}\exp\left[-\frac{(r_1+\lmd)^2}{\sigma^2}\right]
\frac{r_1+\lmd}{|s_2|}\label{eq:exp1pdf1}.
\end{align}
Integrating $f_{\{1\}}dr_1ds_1$ over $r_1>0$ and $s_{(1)}\in\mbD_{\{1\}}$ gives
$\Prob(\calA=\{1\})=\frac{1}{2}e^{-\lmd^2/\sigma^2}$.
In Appendix~\ref{sec:derivation1}, we provide the results for $A=\vn,\{2\},\{1,2\}$,
and verify that $\Prob(\calA=A)$ indeed sums up to one.
\end{example}

\begin{example}[Orthogonal design]\label{exp:orthogonal}
Suppose $p=n=mJ$, $\Psi=\bfI_p$, and put $W=\sqrt{m}\bfI_p$. In this example, all the groups are of
the same size $m$. It is known that under this setting, the group Lasso ($\alp=2$) is obtained 
by block soft-thresholding the least-squares estimator $\tdbeta=\frn X^\trans y$:
\begin{equation}\label{eq:soft}
\hbeta_{(j)}=\tdbeta_{(j)}\left[1-{\lambda\sqm}/{\|\tdbeta_{(j)}\|}\right]_+, \quad j=1,\ldots,J.
\end{equation}
Assume $\veps \sim \dnorm_n(0,\sigma^2 \bfI_n)$. Then
$\tdbeta_{(j)}\sim \dnorm_m(\beta_{0(j)},(\sigma^2/n)\bfI_m)$ are independent of each other.
The distribution of $(\hgamma_\calA,S,\calA)$ for $\alp=2$, derived in Appendix~\ref{sec:derivation2},
is given by 
\begin{equation}\label{eq:densityorth}
d\mu_A = \prod_{j\in \N_J}\left[({2\pi \sigma^2}/{n})^{-\frac{m}{2}} 
\exp \left\{ -\frac{n}{2\sigma^2}\|(r_j+\lambda \sqm) s_{(j)}-\beta_{0(j)}\|^2\right\} |\det M_{(jj)}| d\theta_{(j)}\right],
\end{equation}
where $d\theta_{(j)}=dr_j \wedge ds_{F(j)}$ if $j\in A$ and $d\theta_{(j)}=ds_{(j)}$ (with $r_j=0$) otherwise.
Here, $F(j)$ is a chosen set of $(m-1)$ free coordinates of $s_{(j)}$, and
$\det M_{(jj)}$ has a closed-form expression \eqref{eq:detMjjexp2}.
In what follows, we expemplify that \eqref{eq:densityorth} is consistent with block soft-thresholding \eqref{eq:soft}.

Since $d\mu_A$ factorizes into a product of $J$ terms,
different groups $(\hgamma_{j},S_{(j)})$ are mutually independent. 
Denote the density in each term on the right side of \eqref{eq:densityorth} by $f_j(r_j,s_{(j)})$, 
which determines the distribution of $(\hgamma_{j},S_{(j)})$. 
If $j\notin A$, letting $r_j=0$ we have
\begin{equation}\label{eq:densityzeroorth}
f_j ds_{(j)} = ({2\pi \sigma^2}/{n})^{-\frac{m}{2}} 
\exp \left[ -\frac{n}{2\sigma^2}\|\lambda \sqm s_{(j)}-\beta_{0(j)}\|^2\right] (\lambda\sqm)^m ds_{(j)}.
\end{equation}
It then follows that
\begin{align*}
\Prob(\hbeta_{(j)}=0) & = \int_{\mbB^m} f_j ds_{(j)}
= \int_{\|z\|\leq \lambda \sqm} \phi_m(z;\beta_{0(j)},\sigma^2\bfI_m/n) dz \\
& = \Prob(\|\tdbeta_{(j)}\|\leq \lambda \sqm),
\end{align*} 
where the last equality comes from the distribution of $\tdbeta_{(j)}$. This result is clearly consistent with
the soft-thresholding rule \eqref{eq:soft}. Next we calculate $\Prob(\hgamma_j > t)$ for $j\in A$.
To simplify our derivation, assume further that $\beta_{0(j)}=0$. Integrating $f_j(r_j,s_{(j)})$ 
over the sphere $\|s_{(j)}\|=1$, 
the marginal density of $\hgamma_j$ is 
\begin{align}\label{eq:densityrjorth}
f_j(r_j) = \frac{\left({n}/{\sigma^2}\right)^{\frac{m}2}}{2^{\frac{m}{2}-1} \cdot\Gamma(m/2)}
(r_j + \lambda \sqm)^{m-1} \exp\left[ -\frac{n}{2\sigma^2}(r_j+\lambda \sqm)^2\right]
\end{align}
for $r_j>0$. It then follows that, for $t\geq 0$,
\begin{equation}\label{eq:gmsoft}
\Prob(\|\hbeta_{(j)}\|> t) = \int_t^{\infty} f_j(r_j) dr_j =\Prob\left\{\|\tdbeta_{(j)}\| > t+\lmd \sqm\right\},
\end{equation}  
which again coincides with the result from soft-thresholding.
See Appendix~\ref{sec:derivation2} for the derivation of \eqref{eq:densityrjorth} and \eqref{eq:gmsoft}.
\end{example}

\begin{example}[Lasso]\label{exp:lasso}
When $\alp=1$ in \eqref{eq:blocklassodef}, the block Lasso reduces to the Lasso with no
group structure. Thus, the result for $\alp=1$ can be deduced by letting $p_j=1$ for all $j$ and $\alp=2$ 
(or any $\alp>1$) in Theorem~\ref{thm:densityhigh}.
In this case, for $j\in\calA$ the subgradient $S_j=\sgn(\hbeta_j)\in \{1,-1\}$ is a function of $\hbeta_j$.
This leads to two special properties of the matrix $T=T(\eta(s),A)$
defined in Lemma~\ref{lm:diffHA} which do not hold in the general case $p_j\geq 2$: 
(i) $T=T(A)$ depends only on $A$, 
(ii) the submatrix $T_{A\bullet}$ is a zero matrix; see Appendix~\ref{sec:derivation3}. Bearing these facts in mind,
one can apply Theorem~\ref{thm:densityhigh} to find the joint distribution of the augmented Lasso,
given by the density
\begin{equation}\label{eq:nformLasso}
f_A(r_A,s) dr_Ads_F= g_n(\tdH_A(r_A,s))
\left|\det \left\{\sqn (X^\trans)^+ [\Psi_A \mid \lambda W_{B} T_{B\bullet}]\right\}\right| dr_A ds_F
\end{equation}
for $(r_A,s)\in\Omg_A$, where $B=\N_p\setminus A$ and $F\subset B$.
Owing to property (i), the Jacobian here does not depend on $s$, which is fundamentally
different from the block Lasso. A geometrical interpretation for (i) is that the space 
$\calM_A$ \eqref{eq:defcalSofA}
for the Lasso is a polyhedron and thus the set of tangent vectors
that forms the columns of $T$ is invariant at each $s\in\calM_A$, while in the block Lasso case
$\calM_A$ is curved with a different tangent space at different points. This gives one of the aspects in which
this work represents a highly nontrivial generalization to the result for the Lasso.

As adopted by \cite{Zhou14},
the augmented Lasso estimator can also be represented by $(\hbeta_\calA,S_\calB,\calA)$,
where $\calB=\N_p\setminus\calA$ is the set of zero components of $\hbeta$.
With the change of variable, $\hbeta_j=\hgamma_j S_j$ for $j\in\calA$, one can easily obtain
the density under this alternative parameterization from \eqref{eq:nformLasso},
which is identical to the joint density in Theorem 2 of \cite{Zhou14} with the choice of
$(X^{\trans}/\sqn)$ as a basis for $\row(X)$. See Appendix~\ref{sec:derivation3} for the technical details.
\end{example}

\section{Applications in Statistical Inference}\label{sec:infer}

In this section, we develop Monte Carlo methods to make inference about $\beta_0$ 
by utilizing the joint density of the augmented block Lasso estimator. 
Recall that we want to test the hypothesis $H_{0,G}: \beta_{0G}=0$ or to
construct confidence regions for $\beta_{0G}$. Without loss
of generality, assume that $G=\calG_j$ for some $j\in\N_J$ so that our goal is to 
infer $\beta_{0(j)}$. Denote the null hypotheses by $H_{0,j}: \beta_{0(j)}=0$ for $j\in\N_J$.

\subsection{Parametric bootstrap}\label{sec:pbootstrap}

Consider inference with an estimator in the form of
$\estpivot=\estpivot(\hbeta,S)\in\R^p$, a mapping of the augmented estimator $(\hbeta,S)$.
One such approach that has drawn recent attention is 
the de-biased Lasso and its generalization to the de-biased group Lasso. 
Given a $p\times p$ matrix $\hTheta=\hTheta(X)$, 
a form of the de-biased estimator may be expressed as 
\begin{equation}\label{eq:debiased}
\estpivot=\hbeta+ \hTheta X^\trans(y-X\hbeta)/n =\hbeta+\lambda \hTheta WS,
\end{equation}
where $(\hbeta,S)$ is either the augmented Lasso or the augmented group Lasso.
Different de-biased estimators have been constructed with different $\hTheta$, which is often
some version of a relaxed inverse of the Gram matrix $\Psi$.
It is usually impossible to obtain the exact distribution of $(\estpivot-\beta_0)$ for a finite sample.
Thus, bootstrap methods have been developed \citep{ZhangCheng16,Dezeure16} with improved
performance compared to asymptotic approximations for the de-biased methods.

Assuming the error distribution is $\dnorm_n(0,\sigma^2\bfI_n)$ with a known $\sigma^2$ for now,
a parametric bootstrap for the estimator $\estpivot$ contains three steps:
\begin{routine}[$PB(\tdbeta,\sigma^2,\lambda)$]\label{alg:bp}
Given $\sigma^2>0$, $\lambda>0$ and a point estimate $\tdbeta\in\R^p$, 
\begin{enumerate}
\item[(1)] draw $\varepsilon^*  \sim \dnorm_n(0,\sigma^2\bfI_n)$ 
and set $y^*=X \tdbeta + \varepsilon^*$;
\item[(2)] solve \eqref{eq:blocklassodef} with $y^*$ in place of $y$ to obtain $\hbeta^*$;
\item[(3)] calculate $S^*$ via \eqref{eq:KKTwhole} and $\estpivot^*=\estpivot(\hbeta^*,S^*)$.
\end{enumerate}
\end{routine}
Choosing a function $h_j: \R^{p_j} \to [0,\infty)$, we estimate its $(1-\siglevel)$-quantile $h_{j,(1-\siglevel)}$
from a large bootstrap sample such that
\begin{align*}
\Prob\left\{ \left. h_j(\estpivot^*_{(j)}-\tdbeta_{(j)})>h_{j,(1-\siglevel)} \right| \tdbeta\right\} = \siglevel.
\end{align*}
Then, a $(1-\siglevel)$ confidence region for $\beta_{0(j)}$ can be constructed in the form of 
\begin{align}\label{eq:regionest}
R_{j}(\siglevel)=\left\{\theta\in\R^{p_j}:h_j(\hb_{(j)}-\theta)\leq h_{j,(1-\siglevel)}\right\}.
\end{align}
By duality the p-value for testing $H_{0,j}$ is approximated by the tail probability
\begin{align}\label{eq:pval}
\Prob\left\{ \left. h_j(\estpivot^*_{(j)}-\tdbeta_{(j)})\geq h_j(\hb_{(j)}) \right| \tdbeta\right\}.
\end{align}
Common choices of $h_j$ include, for example, various norms and $h_j(\theta)=\|X_{(j)}\theta\|$.
Although out of the scope of this paper, the asymptotic validity of \eqref{eq:regionest} and \eqref{eq:pval}
comes from the fact that $(\hb_{(j)}-\beta_{0(j)})$ is an asymptotic pivot with a careful choice of $\hTheta$
\citep{MitraZhang14,vandeGeer13}.

An interesting and key observation is that the joint density of $[\hbeta^*,S^*\mid \tdbeta]$ is explicitly given by 
\eqref{eq:nform} in Theorem~\ref{thm:densityhigh}, with $\tdbeta$ in place of $\beta_0$, through its equivalent representation. 
Denote this density \eqref{eq:density} by 
$f_A(r_A,s; \tdbeta,\sigma^2,\lambda)$ to emphasize
its dependence on $(\tdbeta,\sigma^2,\lambda)$.
In principle, we can use Monte Carlo methods,
such as importance sampling and MCMC, to draw $(\hbeta^*,S^*)$
and obtain a sample of $\estpivot^*=\estpivot(\hbeta^*,S^*)$, which serve
as alternatives to the above bootstrap sampling. Monte Carlo methods may bring computational efficiency
and flexibility compared to parametric bootstrap. In the following, we will demonstrate the efficiency 
of importance sampling in calculating tail probabilities as in \eqref{eq:pval}, 
which is a prominent difficulty for the bootstrap. 
Monte Carlo methods for other applications, including with an estimated error distribution, 
are discussed in Section~\ref{sec:otherMC}.

\subsection{Importance sampling}\label{sec:is}

The following simple fact 
about the parameterization of $\calM_A$ \eqref{eq:defcalSofA},
proved in Appendix~\ref{sec:pfparam}, is useful
for designing proposal distributions in importance sampling.

\begin{lemma}\label{eq:MAparameterization}
Let $\alp\in(1,\infty)$. For each $|A|\leq n$, the manifold $\calM_A$, except for a set of measure zero, 
can be parameterized by $s_F$ such that the index set $F=F(A)$ only depends on $A$. 
\end{lemma}


A consequence of Lemma~\ref{eq:MAparameterization} is that we may use the same volume element
$d\theta = dr_A \wedge ds_F$ almost everywhere in the subspace $\Omg_A$,
which eases our development of a Monte Carlo algorithm.
Suppose that $q_A(r_A,s)$ is the density of a distribution over $\Omg$ with respect to $d\theta$
such that $\sum_A \int_{\Omg_A}q_A(r_A,s) d \theta=1$. 
As long as the support of $q_A$ is $\Omg_A$ for all $A\in\scrA$, it can be used as a proposal
distribution in importance sampling.
With a little abuse of notation,
put $\theta=(r_A,s)\in \Omg_A$ so that $(\theta,A)$ represents a point in the sample space $\Omg$
at which the volume element is $d\theta$.
Suppose we want to estimate the expectation of a function 
$h(\hbeta,S)=h(\hgamma,S,\calA)$ with respect to $f_A$, using $(\hbeta,S)$ and
$(\hgamma,S,\calA)$ interchangebly. 
Importance sampling can be readily implemented given the densities $f_A$ and $q_A$.
Draw $(A_t,\theta_t)$ from the proposal $q_{A}(\theta)$ for $t=1,\ldots,N$ and calculate importance weights
$w_t={f_{A_t}(\theta_t)}/{q_{A_t}(\theta_t)}$.
Then by the law of large numbers, the weighted sample mean
\begin{align*}
\hat{h}=\frac{\sum_{t=1}^N w_t h(\theta_t,A_t)}{\sum_{t=1}^N w_t} &\toas
\E[h(\hgamma,S,\calA)]
\end{align*}
provides the desired estimate. To estimate the probability in \eqref{eq:pval}, $h$ is taken to be the indicator
function of the event of interest.
When the true $\beta_{0(j)}\ne 0$, the p-value \eqref{eq:pval} can be tiny, 
and bootstrap (Algorithm~\ref{alg:bp}) may fail to provide a meaningful estimate of the significance level. 
In such cases, it is much more efficient to use importance sampling with a proposal distribution 
that has a higher chance to reach the tail of the bootstrap distribution
$f_A(r_A,s; \tdbeta, \sigma^2,\lambda)$. 

We design two types of proposal distributions. 
The first type of proposals draw $(\hbeta^*,S^*)$ by
the bootstrap algorithm $PB(\beta^{\dag}, M\sigma^2,\lambda^\dag)$ 
with a proper choice of $(\beta^{\dag}, M,\lambda^\dag)$, where $M>0$ is a constant.
The proposal distribution has density $f_A(r_A,s; \beta^{\dag}, M\sigma^2,\lambda^\dag)$,
again by Theorem~\ref{thm:densityhigh}.
By increasing the error variance with $M>1$, choosing $\beta^{\dag}\ne \tdbeta$,
and possibly with a different $\lambda^\dag$, we can propose samples in the region of interest
in \eqref{eq:pval} which has a small probability with respect to the target distribution.
The Jacobian term $J_A(r_A,s;\lambda)$ \eqref{eq:jacob} is the time-consuming part in 
evaluating the densities for calculating importance weights.
If we choose $\lambda^\dag=\lambda$, however, this term will cancel out and the importance weight
is simply the ratio of two normal densities, whose calculation is almost costless.  
Our empirical study shows that this choice gives comparable estimation accuracy and thus
we always let $\lambda^\dag=\lambda$. 
Denote by $IS(\beta^\dag,M)$ the importance sampling with the first type of proposals.
Our second design uses a mixture of two proposal distributions with different $\beta^\dag$ and $M$,
which has more flexibility in shifting samples to multiple regions of interest. 
Again the Jacobian terms cancel out in the importance weight \eqref{eq:normalratio}. 
Our importance sampling with a mixture proposal is detailed in the following algorithm.
For brevity, write
\begin{equation*}
\tdH(\hbeta,S;\beta_0)=\sqn(X^{\trans})^+(\Psi\hbeta+\lambda W S-\Psi \beta_0),
\end{equation*}
which is identical to the $\tdH$ in \eqref{eq:bieps}.

\begin{routine}[($IS(a_1,\beta^\dag_1,M_1;a_2,\beta^\dag_2,M_2)$)]\label{alg:is}
Given $a_1+a_2=1$, $\beta^{\dag}_1,\beta^{\dag}_2\in\R^p$ and $M_1,M_2>0$,
\begin{enumerate}
\item[(1)] draw $Z$ from $\{1,2\}$ with probabilities $\{a_1,a_2\}$ and
$(\hbeta^*,S^*)$ from $PB(\beta^\dag_Z,M_Z\sigma^2,\lambda)$;
\item[(2)] calculate importance weight
\begin{equation}\label{eq:normalratio}
W^*=\frac{\phi_n\left(\tdH(\hbeta^*,S^*;\tdbeta);\sigma^2/n\right)}
{\sum_{k=1}^2 a_k\,\phi_n\left(\tdH(\hbeta^*,S^*;\beta^{\dag}_k);M_k\sigma^2/n\right)}.
\end{equation}
\end{enumerate}
\end{routine}

\begin{remark}
The first algorithm $IS(\beta^\dag,M)$ 
can be regarded as a special case of Algorithm~\ref{alg:is} with $a_1=1$, $\beta^\dag_1=\beta^\dag$
and $M_1=M$. One can easily generalize Algorithm~\ref{alg:is} to a mixture proposal
with $K\geq 3$ component distributions. For other error distributions, we simply replace
$\phi_n$ in  \eqref{eq:normalratio} by $g_n$, the density of $\veps/\sqn$.
\end{remark}

In our numerical results, 
the efficiency of an importance sampling estimate is measured by its coefficient of variation (cv) across
multiple independent runs and compared with direct bootstrap outlined in Algorithm~\ref{alg:bp}.

\subsection{Group Lasso}\label{sec:grpinfer}

We begin with a simpler application to test the complete null hypothesis $H_0:\beta_0=0$
using the statistic $T=h(\hbeta)=\sum_j \|\hbeta_{(j)}\|$, where $\hbeta$ is the group Lasso for a particular $\lambda$. 
In this case, our target density $f_A(r_A,s; \beta_0=0, \sigma^2,\lambda)$ determines the exact
distribution of $T$ under $H_0$.

We set the group size $p_j=10$ for all groups and fixed $\sigma^2=1$.
Each row of $X$ was drawn from $\dnorm_p(0,\Sigma)$, where the diagonal elements of $\Sigma$ are all 1.
The off-diagonal elements $\Sigma_{ij}=\rho_1$ if $i,j$ are in the same group and  $\Sigma_{ij}=\rho_2$ otherwise.
We simulated 30 datasets with parameters $(n,p,\rho_1,\rho_2)$ reported in Table~\ref{tab:compdata}.
Put $v=(1,1,1,1,-1,-1,-1,-1,0,0)$. For the first 10 datasets, we chose $\beta_0=0$ so that $H_0$ was true.
For the other 20 datasets, the first two groups of $\beta_0$ were active, with $\beta_{0(1)}=\beta_{0(2)}=v/2$
for datasets 11 to 20 and $\beta_{0(1)}=\beta_{0(2)}=v$ for datasets 21 to 30.
For each dataset, $\lambda$ was chosen to be the smallest value such that the group Lasso solution had
two active groups. The range of $\lambda$ and that of the statistic $T$ across the simulated datasets are reported
in Table~\ref{tab:compdata} as well.

\begin{table}[ht]
\caption{Simulated datasets for testing complete null hypothesis\label{tab:compdata}} \smallskip
\centering
\begin{tabular}{lllcc}
\hline \hline
Dataset & $(n,p)$ & $(\rho_1, \rho_2)$		& $\lambda$ 		& $T$ \\ \hline
1-10		&  $(30, 100)$ & $(0, 0)$			& $(0.396, 0.796)$	& $(0.017, 0.337)$ \\
11-20 	&  $(30, 100)$ & $(0, 0)$			& $(0.554, 1.613)$	& $(0.460, 1.964)$ \\
21-30 	&  $(30, 100)$ & $(0.5, 0)$		& $(0.956, 2.650)$	& $(0.045, 2.186)$ \\ \hline
\end{tabular} 
\end{table}

We applied the algorithm $IS(0,5)$ to generate $N=100,000$ samples. This procedure was repeated 20 times
independently for each dataset to calculate the mean $\bar{q}=\E(\hat{q}^{(IS)})$
and the standard deviation of an estimated p-value $\hat{q}^{(IS)}$, 
from which we calculated $\text{cv}(\hat{q}^{(IS)})$. 
If we had used the bootstrap algorithm $PB(0,\sigma^2,\lambda)$ for the same $N$ to estimate the p-value,
denoted by $\hat{q}^{(PB)}$, 
its cv would have been close to $\sqrt{(1-\bar{q})/(N\bar{q})}$. Figure~\ref{fig:compgrp} plots
$\log_{10}(\bar{q})$, $\text{cv}(\hat{q}^{(IS)})$ and $\log_{10}\{\text{cv}(\hat{q}^{(PB)})/\text{cv}(\hat{q}^{(IS)})\}$
for the 30 datasets.
We observe from the ratio of cv's in panel (c) 
that the importance sampling estimates are much more accurate for datasets 11 to 30 in which
the p-values are very small. For many of these 20 datasets, the improvement of importance sampling
 over bootstrap can be five or more orders of magnitude. 
The p-values were insignificant for the first 10 datasets in which the null hypothesis was true. 
For a majority of these cases, the importance
sampling estimates were slightly less accurate than the bootstrap estimates, which is fully expected.

\begin{figure}[ht]
  \centering
  \rotatebox{-90}{
    \includegraphics[width=0.4\textwidth]{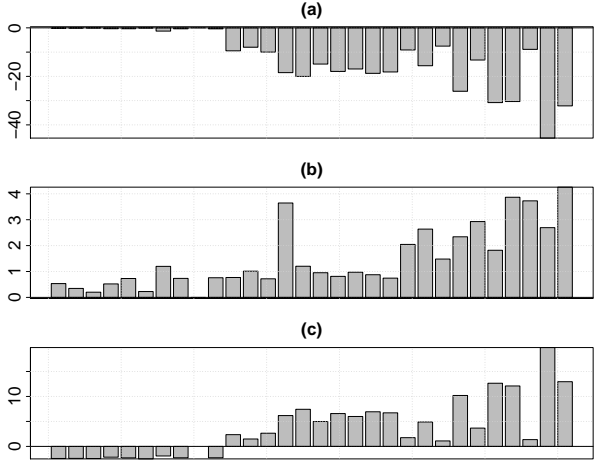}
    }
  \caption{Estimation of p-values for testing $H_0$ with the group Lasso. 
  (a) $\log_{10} \bar{q}$, (b) cv$(\hat{q}^{(IS)})$ and 
  (c) $\log_{10}\{\text{cv}(\hat{q}^{(PB)}) / \text{cv}(\hat{q}^{(IS)})\}$. 
  The result for a dataset is reported by a vertical bar in each plot.
  \label{fig:compgrp}}
\end{figure}

\subsection{A de-biased approach}\label{sec:dBLinfer}

The second application concerns a de-biased group Lasso in the form of \eqref{eq:debiased}.
Since our method applies to any choice of $\hTheta$,
to simplify the discussion we set $\hTheta=\Sigma^{-1}$ instead of
using a particular estimate, where $\Sigma$ is the population covariance of $X$. 
The test statistic is chosen as $h_j(\hb_{(j)})=\|X_{(j)}\hb_{(j)}\|:=T_j$ in \eqref{eq:pval}.

We simulated 20 datasets independently under the same settings as those for 
datasets 11 to 30 in Table~\ref{tab:compdata}. The tuning parameter $\lambda$ was chosen by the same method
as in Section~\ref{sec:grpinfer} to calculate the group Lasso $\hbeta$ 
and the de-biased estimate $\hb$ \eqref{eq:debiased} for each dataset. 
Figure~\ref{fig:dbestimate} plots these two estimates for one dataset, 
in which $\beta_{0(1)}=\beta_{0(2)}=v$ and $\beta_{0(j)}=0$ for $j> 2$.
We see that the de-biased group Lasso $\hb$ is not sparse, $\hb_{(j)}\ne 0$ for all $j$, and its first two groups are closer
to the active groups of $\beta_{0}$ than the group Lasso. This largely removed the shrinkage in the active coefficients
of the group Lasso solution and substantially reduced its bias.
Our goal was then to test $H_{0,1}: \beta_{0(1)}=0$ by 
estimating the probability \eqref{eq:pval} for $T_1=\|X_{(1)}\hb_{(1)}\|$ 
with a plug-in point estimate $\tdbeta$.
The observed value of the test statistic $T_1$ ranged from $4.4$ to $21.2$ across the 20 datasets.
Due to the asymptotic unbiasedness of $\hb_{(j)}$, the bootstrap distribution $[\hb^*_{(j)}-\tdbeta_{(j)}\mid \tdbeta]$ 
is not sensitive to the choice of $\tdbeta$ as long as it is sparse. Thus, we chose $\tdbeta=\hbeta$.
See \cite{Dezeure16} for related discussions.

\begin{figure}[ht]
  \centering
    \includegraphics[height=0.7\textwidth,angle=-90,trim=0.5in 0in 0.2in 0in,clip]{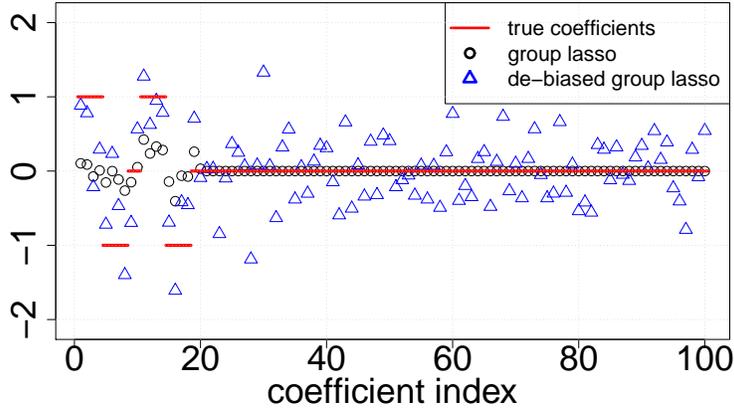}
  \caption{The group Lasso and de-biased group Lasso solutions for one dataset with $p=100$, where
  the size of each group is 10.
  \label{fig:dbestimate}}
\end{figure}

We designed the following mixture proposal for Algorithm~\ref{alg:is} to approximate the p-value \eqref{eq:pval}
by importance sampling:
\begin{align*}
a_1=a_2=1/2; M_1=2, M_2=4; \beta^\dag_{1}=\hbeta,
\beta^\dag_{2(1)}=\hbeta_{(1)}/2, \beta^\dag_{2(-1)}=\hbeta_{(-1)}.
\end{align*}
Note that $\beta^\dag_{2(1)}$ is the middle point between $\hbeta_{(1)}$ and $0$,
serving as a bridge between the target distribution and the null hypothesis $H_{0,1}$.
To achieve a wider coverage of the sample space,
the error variances of both component distributions were chosen to be greater than $\sigma^2$.
We applied Algorithm~\ref{alg:is} to generate $N=100,000$ weighted samples for each dataset,
and replicated this procedure 20 times independently to calculate cv as we did in previous example.
The same comparisons were conducted and reported in Figure~\ref{fig:compdbl}.
Strong majority of the p-values were estimated to be significant, since $\beta_{0(1)}\ne 0$ for all 20 datasets.
The cv's of the importance sampling estimates are seen to be quite small, 
which is especially satisfactory for those
tiny tail probabilities on the order of $10^{-10}$ or smaller.
As shown in Figure~\ref{fig:compdbl}(c), our importance sampling estimation is more efficient than 
parametric bootstrap for at least 13 out of the 20 datasets, many showing orders of magnitude of improvement.
For most of the other datasets, the importance sampling results are very comparable to
the results from bootstrap.

\begin{figure}[t]
  \centering
  \rotatebox{-90}{
    \includegraphics[width=0.4\textwidth]{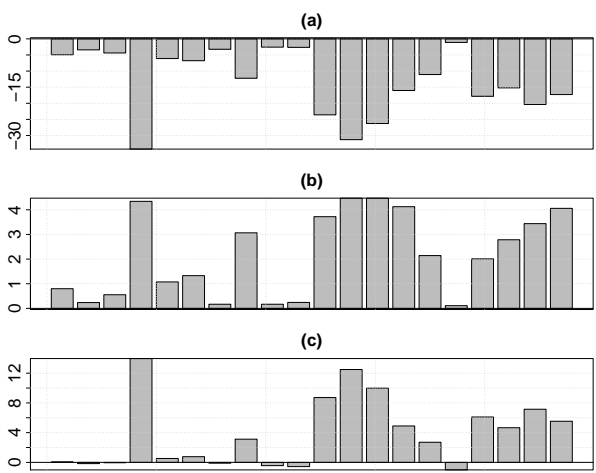}
    }
  \caption{Estimation of p-values for testing $H_{0,1}$ with a de-biased group Lasso. 
  Plots are in the same format as those in Figure~\ref{fig:compgrp}.
  \label{fig:compdbl}}
\end{figure}

Compared to the parametric bootstrap in Algorithm~\ref{alg:bp},
the only additional step in our importance sampling algorithms is to evaluate importance weights,
such as \eqref{eq:normalratio}, of which the computing time is negligible relative to drawing group Lasso
samples. As a result, the total running time of importance sampling is almost identical
to that of the bootstrap sampling. The above two applications thus exemplify the huge gain in estimation
accuracy by importance sampling via estimator augmentation at almost identical computing cost.
It is worth mentioning that accurate estimation of small p-values is crucial for ranking the importance of predictors
and controlling false discoveries in large-scale screening.

\subsection{Other applications}\label{sec:otherMC}

Given the joint density $f_A(r_A,s; \tdbeta,\sigma^2,\lambda)$, one may design
MCMC algorithms to draw samples $(\hbeta^*,S^*)$ from this distribution, which is identical to the distribution of
a bootstrap sample generated by $PB(\tdbeta,\sigma^2,\lambda)$ in Algorithm~\ref{alg:bp}.
The advantage of an MCMC algorithm is that it does not need to solve a convex optimization program
in any of its steps.
But evaluating the Jacobian term in $f_A$ could be time-consuming. 
Another potential application is conditional sampling from $[\hbeta^*,S^* \mid \hbeta^* \in B]$, which may be
useful in post-selection inference. 
For example, conditioning on
the model selected by $\hbeta$, i.e. $G(\hbeta^*)=G(\hbeta)$, we may wish to sample from an estimator $\hb^*$ with
a nice asymptotic distribution for inference. For this problem, bootstrap may be impractical since
the conditioning event is often a rare event. However, from the joint density one can easily obtain 
the conditional density $\propto f_{G}(r_G,s)$, where $G=G(\hbeta)$, and implement an
MCMC algorithm to draw from this conditional distribution.
In the case of Lasso, \cite{Zhou14} implemented an Metropolis-Hastings sampler for such
conditional sampling. The more general case for a block Lasso will be considered in the future.

Under a Gaussian error assumption, it is a common practice to plug an estimated variance $\hsigma^2$ 
in the bootstrap $PB(\tdbeta,\hsigma^2,\lambda)$.
As long as $\hsigma^2$ is consistent with a certain rate, inference will be valid 
asymptotically \citep{Dezeure16,ZhouMain15}.
Therefore, we can use our importance sampling algorithms with 
$f_A(r_A,s; \tdbeta,\hsigma^2,\lambda)$ as the target density.
Note that the density $f_A$ \eqref{eq:density} depends on the error distribution only through
the density $g_n$ of $\veps/\sqn$. Under a general i.i.d. error assumption, estimating $g_n$ reduces
to estimating the density of an univariate distribution, which can be done quite accurately even when
$n$ is moderate by either a parametric or a nonparametric method. Given an estimate $\hat{g}_n$,
our target density is readily obtained with $g_n$ replaced by $\hat{g}_n$.

\section{Discussions}\label{sec:generalization}

\subsection{Block-$(1,\infty)$ norm}\label{sec:inftynorm}

In this subsection, we consider the case $\alp=\infty$ ($\alp^*=1$). 
The difference between this case and the case $\alp<\infty$ comes from the subgradient vector $S$. 
Let $\calB_j=\argmax_{k\in\calG_j}|\hbeta_k|\subset \calG_j$,
which may contain multiple elements when a tie occurs,
and $\calB_j^c=\calG_j\setminus \calB_j$ for $j\in \N_J$.
It follows from Lemma~\ref{lm:subdiffLinfty} that (i) 
for $\hbeta_{(j)}\ne 0$, $\|S_{(j)}\|_1= 1$ and
\begin{equation}\label{eq:Sinftyn0}
S_k=\begin{cases}
t_k \sgn(\hbeta_k) & k\in \calB_j \\
0 & k \in \calB^c_j
\end{cases},
\end{equation}
where $\sum_{\calB_j}t_k=1$ and $t_k\geq 0$; (ii) $\|S_{(j)}\|_1\leq 1$ for $\hbeta_{(j)}=0$.
 
Compared to \eqref{eq:Sdef}, the discreteness in $\{S_k=0\}$ for some $k$ as in \eqref{eq:Sinftyn0} distinguishes 
the $(1,\infty)$-norm from other cases of $\alp<\infty$.
Accordingly, the augmented estimator $(\hbeta,S)$ will have more structure. 
Recall that the active blocks of $\hbeta$ are denoted by $\calA=G(\hbeta)$. 
For $j\in\calA$, define 
\begin{equation}\label{eq:calKj}
\calK_j=\{k\in\calG_j: S_k\ne 0\}\quad\text{and}\quad \calK^c_j=\calG_j\setminus \calK_j.
\end{equation}
Put $\calK=\cup\{\calK_j:j\in\calA\}$ and $\calK^c=\cup\{\calK^c_j:j\in\calA\}$.
It follows from \eqref{eq:Sinftyn0} that
$\hbeta_{\calK_j}=\hgamma_j \sgn(S_{\calK_j})$ for $j\in\calA$, where $\hgamma_j=\|\hbeta_{(j)}\|_{\infty}$.
We can then represent $(\hbeta,S)$ by 
\begin{equation}\label{eq:equivrep}
(\hgamma_\calA,\hbeta_{\calK^c},S, \calA, \calK)\quad \text{with } S_{\calK^c}=0, 
\end{equation}
subject to the constraints that 
$\|\hbeta_{\calK_j^c}\|_{\infty}\leq \hgamma_j$ for  $j\in\calA$.
For $\calA=A\in\scrA$ \eqref{eq:Arange}, Proposition~\ref{prop:uniqsuffinfty} implies that 
assuming solution uniqueness the range of $\calK$ is 
\begin{equation}\label{eq:rangeofK}
\scrK(A)=\{K\subset \calG_A: K\cap \calG_j \ne \vn, \forall j\in A \text{ and } |\calG_A\setminus K|\leq n-|A|\}.
\end{equation}
Let $K_j=K\cap \calG_j$ and $K^c_j=K^c\cap \calG_j$ for $j\in A$, where $K^c=\calG_{A}\setminus K$.
The sample space for $S$ given $\calA=A$ and $\calK=K$ is
\begin{equation}\label{eq:defcalSofAinfty}
\calM_{A,K}=\left\{s\in  \calV: \|s_{K_j}\|_{1}=1,s_{K_j^c}=0\;\forall j\in A 
\text{ and } \|s_{(j)}\|_{1}\leq 1\;\forall j\notin A\right\},
\end{equation}
where $\calV$ is defined in \eqref{eq:SinrowX}.  
The sample space for $(\hgamma_j,\hbeta_{K_j^c})$ is the cone 
\begin{equation}\label{eq:cone}
\calC_j=\left\{(r,v)\in (\R^+)\times \R^{|K_j^c|}: \|v\|_{\infty} \leq r_j\right\}.
\end{equation}
Then the sample space for $(\hgamma_A,\hbeta_{K^c},S)$ is the product
$\Omega_{A,K}=(\prod_{j\in A} \calC_j) \times \calM_{A,K}$.
Obviously, taking union over the range of the sets $A\in\scrA$ and $K\in\scrK(A)$ determines
the space $\Omega$ for the augmented estimator. 
Compared to the case $\alp<\infty$, the subgradient $S$ has lost $|\calK^c|$ 
free dimensions due to the constraints that $S_{k}=0$ for all $k\in\calK^c$.
Consequently, for every interior point $s\in\calM_{A,K}$,
there is a neighborhood that may be parameterized by $s_F$ with $|F|=(n-|A|-|K^c|):=q$.
Let $I=\N_p\setminus K^c$ and note that $ds_k=0$ for each $k\in K^c$.
Similar to Lemma~\ref{lm:diffHA}, we can then find a matrix $T \in \R^{(p-|K^c|) \times q}$
such that $ds_I=T ds_F$. 

For notational brevity we will use $(\hbeta,S)$ and its equivalent representation \eqref{eq:equivrep} interchangeably.
Write the mappings $H$ \eqref{eq:defineH} and $\tdH$ \eqref{eq:bieps} 
as $H(b,s)$ and $\tdH(b,s)$, respectively, where $(b,s)$ denotes the value of $(\hbeta,S)$. 
For $(r_A,b_{K^c},s)\in\Omg_{A,K}$,
let $\tdH_{A,K}(r_{A}, b_{K^c}, s)=\tdH(b,s)$ with $(r_A,b_{K^c},s,A,K)$ being the equivalent representation of $(b,s)$.
Define two matrices
\begin{align*}
\Psi\circ \sgn(s) & =[\Psi_{(1)}\sgn(s_{(1)})|\ldots | \Psi_{(J)}\sgn(s_{(J)})] \in \R^{p\times J} \\
M(s,A,K;\lmd)& = \left[\{\Psi\circ \sgn(s)\}_A \mid \Psi_{K^c}\mid \lambda {W_I}T\right] \in \R^{p\times n},
\end{align*}
and a related Jacobian 
\begin{equation}\label{eq:jacobinfty}
J_{A,K}(s;\lambda)=\det \left[\sqn (X^\trans)^+ M(s,A,K;\lmd)\right].
\end{equation}
Parallel to Theorem~\ref{thm:densityhigh}, We have the following explicit density for the augmented
estimator under block-$(1,\infty)$ sparsity, which is proved in Appendix~\ref{sec:proofdeninf}.

\begin{theorem}\label{thm:densityinfty}
Fix $p\geq n$, $\beta_0\in\R^p$ and $\lambda>0$. 
Suppose Assumption~\ref{as:X} holds and 
that the program~\eqref{eq:blocklassodef} for $\alp=\infty$ has a unique solution for almost all $y\in\R^n$.
Let $g_n$ be the density of $(\veps/\sqn)$.
Then the distribution of the augmented estimator $(\hbeta,S)$ is given by the $n$-form
\begin{align}\label{eq:nforminfty}
d\mu_{A,K}&=\Prob(d r_{A}, d b_{K^c}, ds,\{A,K\}) \nonumber\\
&= g_n(\tdH_{A,K}(r_{A}, b_{K^c}, s;\beta_0,\lambda))|J_{A,K}(s;\lambda)| d\theta 
:= f_{A,K}(r_A,b_{K^c},s)d\theta
\end{align}
for $(r_A,b_{K^c},s,A,K)\in \Omg$, where $\theta=(r_A,b_{K^c},s_F)\in\R^n$.
\end{theorem}

The sufficient condition for solution uniqueness in this case is discussed in Appendix~\ref{sec:proofsinfty}.
The density $f_{A,K}$ is defined in terms of the parameterization $\theta$.
Suppose that $\Gamma\subset \R^{|A|+|K^c|}$ is a subset of the product cone  $\prod_{j\in A}{\calC_j}$ and
$\Phi=\{\Phi(s_F):s_F\in\Delta\}\subset \calM_{A,K}$ is a $q$-surface in $\R^p$
with parameter domain $\Delta\subset \R^q$.
Then we have
\begin{equation}\label{eq:probabilityinfty}
\Prob\left\{(\hgamma_A, \hbeta_{K^c})\in \Gamma, S\in \Phi, \calA = A, \calK=K\right\}
=\int_{\Gamma \times \Delta} f_{A,K}(r_A,b_{K^c},\Phi(s_F)) d\theta,
\end{equation}
which interprets the differential form \eqref{eq:nforminfty}.
The density here differs from that in \eqref{eq:nform}
only in the Jacobian term. Clearly, the same importance sampling method (Algorithm~\ref{alg:is})
can be used here due to the cancellation of the Jacobian.

\subsection{Concluding remarks}

By augmenting the sample space to that of $(\hbeta,S)$, we have derived a closed-form
density for the sampling distribution of the augmented block Lasso estimator. 
Given the density, we have demonstrated the use of importance sampling in group inference,
which can be orders of magnitude more efficient than the corresponding parametric bootstrap.
For high-dimensional data, sparsity seems an essential assumption for inference,
and consequently, an inference method is often built upon a non-regular penalized estimator.
It is unlikely to work out an exact pivot in this
setting, and thus, simulation-based approaches have been widely used. Our work of estimator augmentation
opens the door to a large class of Monte Carlo methods for such simulations,
which in our view is the main intellectual contribution.
Due to the complexity of the sample space of an augmented estimator, development of
efficient Monte Carlo algorithms is a highly demanding job and an interesting future direction. 


\renewcommand{\theremark}{\thesection.\arabic{remark}}
\setcounter{remark}{0}

\myappendix
\section{Uniqueness of The Block Lasso}\label{sec:app_uniqueness}

\subsection{Auxiliary lemmas}\label{sec:prooflemmas}

\begin{lemma}\label{lm:rhoStoS}
If $\alp\in(1,\infty)$, then $\eta$ is a bijection that maps $\mbS_{\alp^*}$ onto $\mbS_{\alp}$
and $\langle \eta(v), v \rangle=1$ for any $v\in \mbS_{\alp^*}$.
\end{lemma}

\begin{proof}
 For any $v=(v_i)\in \mbS_{\alp^*}$,
\begin{align*}
\|\eta(v)\|_{\alp}^{\alp}=\|(v_i^\rho)\|_{\alp}^{\alp}=\sum_i |v_i|^{\alp^*}=1.
\end{align*}
Similarly, we can show that $\eta^{-1}(u)\in\mbS_{\alp^*}$ for any $u\in\mbS_{\alp}$.
By definition, $\rho+1=\alp^*$. Then, straightforward calculation leads to
\begin{align*}
\la\eta(v), v \ra =\la \sgn(v) |v|^\rho, \sgn(v) |v|\ra = \sum_i |v_i|^{\rho+1} =\sum_i |v_i|^{\alp^*} =1.
\end{align*}
Here, $|\cdot|$ and $\sgn(\cdot)$ are applied on $v$ in the sense of \eqref{eq:fonvec}.
\end{proof}

\begin{lemma}\label{lm:subdiff}
Let $h(v)=\|v\|_{\alp}$ for $\alp\in(1,\infty)$ and $v\in\R^m$. 
If $v\ne 0$, then 
\begin{equation}\label{eq:gradh}
\grad h(v)=\eta^{-1}(\tdv)\in\mbS^{m-1}_{\alp^*}
\end{equation} 
with $\tdv=v/\|v\|_{\alp} \in \mbS^{m-1}_{\alp}$. If $v=0$, the subdifferential of $h$
\begin{equation}\label{eq:subgradh}
\partial h(0)=\{u \in \R^m:\|u\|_{\alp^*}\leq 1\}.
\end{equation}
\end{lemma}

\begin{proof}
For $v\ne 0$, 
\begin{align*}
\frac{\partial h}{\partial v_i} =\frac{\sgn(v_i)|v_i|^{\alp-1}}{\|v\|_{\alp}^{\alp-1}}
=\sgn(\tdv_i) |\tdv_i|^{1/\rho} = \eta^{-1}(\tdv_i),
\end{align*}
using the simple fact that $\alp-1=1/\rho$. 
Since by definition $\tdv\in\mbS_{\alp}$, Lemma~\ref{lm:rhoStoS} implies that 
$\eta^{-1}(\tdv)\in\mbS_{\alp^*}$. This proves \eqref{eq:gradh}. By H\"{o}lder's inequality,
\begin{equation*}
\la u, v\ra \leq \|u\|_{\alp^*} \|v\|_{\alp} \leq h(v), \quad \forall v\in\R^m,
\end{equation*}
if and only if $\|u\|_{\alp^*}\leq 1$, which implies \eqref{eq:subgradh}.
\end{proof}

\begin{lemma}[\cite{Negahban11}, Lemma 1]\label{lm:subdiffLinfty}
Let $h(v)=\|v\|_{\infty}$ for $v\in \R^m$ and $K=\argmax_{i}|v_i|\subset \N_m$.
For $v\ne 0$, $u\in \partial h(v)$ if and only if
\begin{equation*}
u_i=\begin{cases}
t_i\sgn(v_i) & i\in K \\
0 & \text{otherwise}
\end{cases}
\end{equation*}
for some $(t_i)_{i\in K}$ so that $\sum_i t_i=1$ and $t_i\geq 0$.
For $v=0$, 
\begin{equation}\label{eq:subdiffLinftyatzero}
\partial h(0)=\{u\in \R^m: \|u\|_1 \leq 1\}.
\end{equation}
\end{lemma}


\subsection{Characterization of solutions}\label{sec:characterize}

\begin{proof}[Proof of Lemma~\ref{lm:fittedandS}]
Suppose that $\hbeta^{(1)}$ and $\hbeta^{(2)}$ are two minimizers of $L(\beta;\alp)$ such that
$X\hbeta^{(1)}\ne X\hbeta^{(2)}$. The convexity of $L$ implies that $L(\hbeta^{(1)};\alp)=L(\hbeta^{(2)};\alp)=L^*$.
Since $\|x\|^2$ is strictly convex in $x$, for any $c\in(0,1)$
\begin{equation*}
\|y-X [c\hbeta^{(1)} +(1-c)\hbeta^{(2)}] \|^2 
< c \|y-X\hbeta^{(1)}\|^2+(1-c) \|y-X\hbeta^{(2)}\|^2
\end{equation*}
by the hypothesis that $X\hbeta^{(1)}\ne X\hbeta^{(2)}$. Therefore,
\begin{align*}
L(c\hbeta^{(1)}+(1-c)\hbeta^{(2)};\alp)< c L(\hbeta^{(1)};\alp)+(1-c)L(\hbeta^{(2)};\alp) = L^*,
\end{align*}
which is contradictory to the assumption that the minimum of $L$ is $L^*$.
The uniqueness of $S$ is an immediate consequence of the uniqueness of $X\hbeta$
and that
\begin{equation}\label{eq:solveS}
S=(n\lambda{W})^{-1}X^{\trans}(y-X\hbeta)
\end{equation}
by the KKT conditions \eqref{eq:KKTwhole}.
\end{proof}

We will first establish sufficient conditions for solution uniqueness for $\alp<\infty$,
while deferring the case $\alp=\infty$ to Appendix~\ref{sec:proofsinfty}.
We start with more explicit expressions for $X\hbeta$ and $\hbeta$. 
Write the KKT condition for each block in \eqref{eq:KKTwhole},
\begin{equation}\label{eq:KKT}
\frn X_{(j)}^{\trans} X \hbeta + \lambda w_j S_{(j)} = \frn X_{(j)}^{\trans} y, \quad j=1,\ldots,J.
\end{equation}
Define
\begin{equation}\label{eq:equicor}
\calE=\left\{j\in\N_J:\frac{1}{w_j n} \left\|X_{(j)}^{\trans}(y-X\hbeta)\right\|_{\alp^*}
=\lambda\|S_{(j)}\|_{\alp^*}=\lambda\right\},
\end{equation}
and by \eqref{eq:KKT} and \eqref{eq:Sdef}, $\hbeta_{(-\calE)}=0$.
Now the $\calE$ block of \eqref{eq:KKT} with $\hbeta_{(-\calE)}=0$ reads
\begin{equation}\label{eq:KKTBblock}
\frn X_{(\calE)}^{\trans}(y-X_{(\calE)}\hbeta_{(\calE)})=\lambda {W}_{(\calE\calE)}S_{(\calE)},
\end{equation}
which shows that ${W}_{(\calE\calE)}S_{(\calE)}\in \row(X_{(\calE)})$. Thus, we have
\begin{equation}\label{eq:exprofS}
{W}_{(\calE\calE)}S_{(\calE)}=X_{(\calE)}^\trans (X_{(\calE)}^\trans)^+ {W}_{(\calE\calE)}S_{(\calE)},
\end{equation}
since the right side is the projection of ${W}_{(\calE\calE)}S_{(\calE)}$ onto $\row(X_{(\calE)})$.
Plugging the above identity into \eqref{eq:KKTBblock}, we arrive at
\begin{equation}\label{eq:blockmiddle}
 X_{(\calE)}^{\trans}X_{(\calE)} \hbeta_{(\calE)}
 =X_{(\calE)}^{\trans}\left[y-n\lambda (X_{(\calE)}^\trans)^+{W}_{(\calE\calE)}S_{(\calE)}\right].
 \end{equation}
A solution to the above equation is 
\begin{align*}
 \hbeta_{(\calE)}
 &=(X_{(\calE)}^{\trans}X_{(\calE)})^+ X_{(\calE)}^{\trans}\left[y-n\lambda (X_{(\calE)}^\trans)^+{W}_{(\calE\calE)}S_{(\calE)}\right] \\
 &=(X_{(\calE)})^+\left[y-n\lambda (X_{(\calE)}^\trans)^+{W}_{(\calE\calE)}S_{(\calE)}\right].
 \end{align*}
Then by the uniqueness of the fit $X\hbeta$ (Lemma~\ref{lm:fittedandS}), for all solutions $\hbeta$ we have
\begin{equation}\label{eq:fitexpression}
X\hbeta=X_{(\calE)}\hbeta_{(\calE)}=X_{(\calE)}(X_{(\calE)})^+\left[y-n\lambda (X_{(\calE)}^\trans)^+{W}_{(\calE\calE)}S_{(\calE)}\right]
:=\hy.
\end{equation}
Conversely, since 
\begin{equation*}
X_{(\calE)}^\trans X_{(\calE)}X_{(\calE)}^+=X_{(\calE)}^\trans(X_{(\calE)}X_{(\calE)}^+)^\trans
=(X_{(\calE)}X_{(\calE)}^+X_{(\calE)})^\trans=X_{(\calE)}^\trans,
\end{equation*}
\eqref{eq:fitexpression} implies \eqref{eq:blockmiddle} and thus the $\calE$ block
of the KKT conditions \eqref{eq:KKTBblock}. Therefore, \eqref{eq:fitexpression} with $\hbeta_{(-\calE)}=0$
is {\em sufficient and necessary} for $\hbeta$ to be a block Lasso solution.
To make the relation between $\hbeta_{(\calE)}$ and $S_{(\calE)}$ more explicit, write
\begin{equation}\label{eq:etaSj}
\hbeta_{(j)}=\|\hbeta_{(j)}\|_{\alp}\eta(S_{(j)})=\hgamma_j\eta(S_{(j)}) \quad\text{ for } j\in\calE,
\end{equation}
which follows from \eqref{eq:Sdef}.
For $B\subset \N_p$, let $\R^B$ be $|B|$-dimensional Euclidean space with coordinates index by $B$
so that a vector $v\in\R^B$ has components $v_j$, $j\in B$. Similarly, $\R^{m\times B}$ denotes the space
of matrices with columns indexed by $B$. Put
\begin{equation}\label{eq:defgamZ}
\quad Z_j=X_{(j)}\eta(S_{(j)}) \in\R^n, 
\quad Z=(Z_j)_{j\in \calE} \in \R^{n\times \calE}.
\end{equation}
Then \eqref{eq:fitexpression} can be rewritten
\begin{equation*}
\sum_{j\in\calE} \hgamma_j X_{(j)}\eta(S_{(j)})=Z\hgamma_{\calE}=\hy.
\end{equation*}
Now we have the following characterization of the block Lasso solutions:

\begin{lemma}\label{lm:conditionuniq}
A necessary and sufficient condition for $\hbeta$ to be a block Lasso solution \eqref{eq:blocklassodef} 
with $\alp\in(1,\infty)$ is 
\begin{equation}
Z\hgamma_{\calE}=\hy\quad\text{and}\quad\hbeta_{(-\calE)}=0.
\end{equation}
Moreover, $\calE$, $\hy$ and $Z$ are unique for any $y$, $X$ and $\lambda>0$.
\end{lemma}

The uniqueness of $(\hy,\calE,Z)$ is an immediate consequence of Lemma~\ref{lm:fittedandS}.
So non-uniqueness can only come from $\hgamma_{\calE}$ when 
the linear system $Zx=\hy$ has multiple solutions for $x$, which happens only if $\nul(Z)\ne\{0\}$.
Therefore, $\hbeta$ is a block Lasso solution if and only if 
\begin{equation}\label{eq:characterize}
\hbeta_{(-\calE)}=0 \quad\text{and}\quad \hgamma_{\calE}=Z^+ \hy + \gamma, 
\end{equation}
provided that
\begin{equation}\label{eq:condgamma}
\gamma \in \nul(Z)\subset \R^{\calE}\quad\text{and}\quad (Z^+ \hy + \gamma)_j \geq 0\;\text{ for $j\in\calE$.} 
\end{equation}

\subsection{Proof of sufficiency}\label{sec:proofuniq}

If $\nul(Z)=\{0\}$, then $\hbeta$ is uniquely given by \eqref{eq:characterize} with $\gamma=0$.
Note that in this case $\hgamma_j=(Z^+ \hy)_j$ is necessarily nonnegative as 
there always exists a solution to the block Lasso problem. 
Furthermore, $|\calE|\leq n$ and thus this solution has at most $(n\wedge J)$ nonzero blocks.
This leads to our first sufficient condition for the uniqueness of $\hbeta$.

\begin{proposition}\label{prop:uniqsuff}
For any $y$, $X$, and $\lambda>0$, if $\nul(Z)=\{0\}$, 
then the block Lasso solution for $\alp\in(1,\infty)$ is uniquely given by
\begin{equation}\label{eq:uniqcharacterize}
\hbeta_{(-\calE)}=0, \quad\quad \hgamma_{\calE}=(Z^\trans Z)^{-1} Z^\trans \hy,
\quad\text{and}\quad \hbeta_{(j)}=\hgamma_j\eta(S_{(j)}) \text{ for } j\in\calE,
\end{equation}
Furthermore, $|G(\hbeta)|\leq n\wedge J$.
\end{proposition}

In the following, we prove Theorem~\ref{thm:uniqsuff} for $\alp\in(1,\infty)$. Note that the case $\alp=1$
is equivalent to the case $\alp=2$ with $p_j=1$ for all $j$. Thus, this part covers the range
of $\alp\in[1,\infty)$ as in Theorem~\ref{thm:uniqsuff}.
The result for $\alp=\infty$ will be established in next subsection. 

\begin{proof}[Proof of Theorem~\ref{thm:uniqsuff}]
Suppose that $\nul(Z)\ne\{0\}$.
Then for some $i\in \calE$, there is a set $A\subset \calE\setminus\{i\}$ and $|A|\leq n$ such that 
\begin{equation*}
Z_i/w_i = \sum_{j\in A} c_j (Z_j/w_j),
\end{equation*}
where we may assume that $Z_j$, $j\in A$, are linearly independent and $c_j\ne 0$ without loss of generality.
Let $r=y-X\hbeta$ denote the block Lasso residual. By \eqref{eq:KKT}, for every $j\in\calE$,
\begin{align*}
\langle Z_j, r \rangle=\langle X_{(j)}\eta(S_{(j)}),r \rangle = n\lambda w_j \la \eta(S_{(j)}),S_{(j)}\ra
=n\lambda w_j,
\end{align*}
where the last equality follows from Lemma~\ref{lm:rhoStoS} since $S_{(j)}\in\mbS_{\alp^*}$.
Therefore, for $\lambda>0$ we have
\begin{equation*}
1 = \sum_{j\in A} c_j.
\end{equation*}
Note that $Z_j=X_{(j)}\eta(S_{(j)})$ with $S\in\row(XW^{-1})$ and $\|S_{(j)}\|_{\alp^*}=1$ for $j\in \calE$. 
The above equality is thus contradictory to the assumption that the columns of $XW^{-1}$ are in
blockwise general position (Assumption~\ref{as:pj}).
\end{proof}

\subsection{The case of $\alp=\infty$}\label{sec:proofsinfty}

Recall the KKT conditions in \eqref{eq:KKT}. Let $\alp^*=1$ in \eqref{eq:equicor} to define $\calE$.
By definition $\hbeta_{(j)}=0$ for $j\notin \calE$.
For $j\in\calE$ define $\calK_j$ and $\calK^c_j$ as in \eqref{eq:calKj}.
Note that both $\calE$ and $\calK_j$ are unique due to the uniqueness of $S$ and $\hy=X\hbeta$
for any $y$, $X$ and $\lmd>0$ (Lemma~\ref{lm:fittedandS}).
It follows from \eqref{eq:Sinftyn0} and \eqref{eq:calKj} that
$\hbeta_{\calK_j}=\hgamma_j \sgn(S_{\calK_j})$ for each $j\in\calE$.
Then the fitted value $\hy$ can be expressed as
\begin{equation*}
\hy=X_{(\calE)}\hbeta_{(\calE)}=
\sum_{j\in \calE} \left\{\hgamma_j X_{\calK_j}\sgn(S_{\calK_j})+X_{\calK_j^c}\hbeta_{\calK_j^c}\right\}
=Z\hzeta,
\end{equation*}
where we define
\begin{equation}\label{eq:defZz}
Z=\left[X_{\calK_j}\sgn(S_{\calK_j}) \mid X_{\calK_j^c}\right]_{j\in\calE} \quad\text{and}\quad 
\hzeta=(\hgamma_j,\hbeta_{\calK_j^c})_{j\in \calE}.
\end{equation}
If $\nul(Z)=\{0\}$, then $\hzeta$ and hence $\hbeta$ will be unique and 
$Z$ has at most $n$ columns. 
Now we generalize Proposition~\ref{prop:uniqsuff} 
to the block-$(1,\infty)$ norm regularization. 

\begin{proposition}\label{prop:uniqsuffinfty}
Suppose that $\nul(Z)=\{0\}$. Then 
for any $\lambda>0$ and $y\in\R^n$, 
the solution $\hbeta$ to the block Lasso problem \eqref{eq:blocklassodef} with $\alp=\infty$ is uniquely given by 
\begin{equation*}
\hbeta_{(-\calE)}=0, \quad\quad \hzeta=(Z^\trans Z)^{-1} Z^\trans \hy,
\quad\text{and}\quad \hbeta_{\calK_j}=\hgamma_j \sgn(S_{\calK_j}) \text{ for } j\in\calE.
\end{equation*}
Furthermore, $|G(\hbeta)|\leq |\calE|\leq n\wedge J$
and $|\calE|+\sum_{j\in\calE}|\calK_j^c|\leq n \wedge p$.
\end{proposition}


\section{Remaining Proofs}

\subsection{Proof of Lemma~\ref{lm:diffHA}}\label{sec:pfdHA}

Consider the equality constraints on $s$ that are 
involved in the definition of $\calM_A$ \eqref{eq:defcalSofA}.
Let $Q=Q(X)\in \R^{p\times (p-n)}$ be a matrix whose columns form a basis 
for $\calV^\perp=\nul(XW^{-1})$.
By Assumption~\ref{as:X}, $\rank(Q)=p-n$. The equality constraints on $s$ are
\begin{align}
Q^{\trans} s & =0, \label{eq:rowspace}\\
\|s_{(j)}\|_{\alp^*}& =1 \quad\forall  j\in A, \label{eq:sphere}
\end{align}
where \eqref{eq:rowspace}  is equivalent to $s\in \calV$.
As a result, we have the following constraints on $ds$:
\begin{align}
Q^{\trans} ds & =0, \label{eq:rowspaceds}\\
\langle \eta(s_{(j)}), ds_{(j)} \rangle&=0 \quad\forall  j\in A. \label{eq:sphereds}
\end{align}
By Lemma~\ref{lm:rhoStoS}, we have $\langle \eta(s_{(j)}), s_{(j)} \rangle=1$ for each $j\in A$.
Comparing this to \eqref{eq:rowspace}, we see that \eqref{eq:rowspaceds} and \eqref{eq:sphereds}
represent $(p-n+|A|)$ linearly independent constraints and therefore
$ds$ has $n-|A|$ free coordinates. Let $F=F(s,A)\subset \N_p$ be a chosen set 
so that $s_F$ parameterizes $\calM_A$ in a neighborhood of $s$.
Then there exists a matrix $T=T(\eta(s),A)\in \R^{p\times (n-|A|)}$ of rank $n-|A|\leq p$ such that
$ds=T(\eta(s),A) ds_F$ in this neighborhood. 

Fixing $\calA=A$ in \eqref{eq:defineH} leads to the differentiation of $H_A$:
\begin{align*}
d H_A  = \sum_{j\in A}\left[(dr_j)\Psi_{(j)} \eta(s_{(j)})+r_j\Psi_{(j)} d \eta(s_{(j)})\right] + \lambda {W} ds.
\end{align*}
Since $d \eta(s_{(j)})=D_{(jj)}ds_{(j)}$ for $j\in A$, we arrive at
\begin{equation}\label{eq:dHA}
d H_A  =  (\Psi\circ \eta)_A dr_A + \{(r\circ \Psi)D + \lambda W\} ds. 
\end{equation}
Plugging $ds=Tds_F$ into the above and letting $\theta=(r_A,s_F)\in \R^n$ complete the proof.

\subsection{Proof of Theorem~\ref{thm:densityhigh}}\label{sec:proofhimden}

Put $v=\tdH_A(r_A,s)$ for $A\in\scrA$ and $(r_A,s)\in \Omg_A$.
The differential of $\tdH_A$ leads to
\begin{equation}\label{eq:dvdtheta}
dv=\sqn (X^\trans)^+ M(r_A,s,A;\lambda)  d\theta,
\end{equation}
where $M$ and $\theta=(r_A,s_F)$ are as in \eqref{eq:defM}.
Let $\Phi\subset \calM_A$ be a neighborhood of $s$ with parameter domain $\Delta$, 
i.e. $\Phi=\{\Phi(s_F): s_F\in \Delta \subset \R^k\}$, where $k=n-|A|$. 
Suppose that $\Gamma\subset (\R^+)^{|A|}$ is open and contains $r_A$.
Denote by $V\subset\R^n$ the image of $\Gamma\times \Phi$ under $\tdH_A$. 
To establish the $n$-form in \eqref{eq:nform}, it is sufficient to show \eqref{eq:probability} for $u=s_F$.
Under Assumptions~\ref{as:X} and \ref{as:pj}, the bijective nature of $\tdH_A$ \eqref{eq:bieps} implies that 
\begin{align*}
\Prob(\hgamma_A \in \Gamma, S\in \Phi, \calA = A)  =\Prob(\veps/\sqn \in V) 
 = \int_V g_n(v) dv.
\end{align*}
Applying a change of variable in differential forms, we arrive at
\begin{align*}
\int_V g_n(v) dv& =\int_{\Gamma\times \Delta} g_n\left(\tdH_A(r_A,\Phi(s_F))\right)
\left|\frac{\partial v}{\partial \theta}\right| d\theta \\
&= \int_{\Gamma\times \Delta} f_A(r_A,\Phi(s_F)) d\theta,
\end{align*}
where the Jacobian is determined by \eqref{eq:dvdtheta} and $f_A$ is defined in \eqref{eq:density}.
This completes the proof.

\subsection{Proof of Corollary~\ref{cor:dGaussian}}\label{sec:pfGaussian}

If $\veps\sim \dnorm_n(0,\sigma^2\bfI_n)$, then $g_n(v)=\phi_n(v;\sigma^2/n)$.
It follows from \eqref{eq:defineH} and \eqref{eq:bieps} that
\begin{align*}
\tdH_A(r_A,s)=\sqn(X^\trans)^+(\Psi b+\lmd Ws -\Psi\beta_0).
\end{align*}
Since $\Psi^+=nX^+(XX^{\trans})^{-1}X$ by Assumption~\ref{as:X}, we have
$\Psi\Psi^+=X^\trans(XX^{\trans})^{-1}X = P_{X^\trans}$, which
is the projection onto $\row(X)$. Thus $Ws=\Psi\Psi^+Ws$, since $Ws\in\row(X)$.
Putting this together with the identity $(X^\trans)^+\Psi=X/n$, we have
\begin{align*}
\tdH_A(r_A,s)=\frac{1}{\sqn} X(b+\lmd \Psi^+Ws-\beta_0),
\end{align*}
and thus
\begin{equation}\label{eq:gforGaussian}
g_n(\tdH_A(r_A,s))=\left(\frac{2\pi\sigma^2}{n}\right)^{-n/2}
\exp\left[-\frac{1}{2\sigma^2}\|X(b+\lmd \Psi^+Ws-\beta_0)\|^2\right].
\end{equation}
Then \eqref{eq:dGaussian} follows immediately.

\subsection{Proof of Corollary~\ref{cor:densitylow}}\label{sec:pflowden}

It is easy to see that Assumptions~\ref{as:X} and \ref{as:pj} hold trivially if $\rank(X)=p<n$. 
Thus, by Lemma~\ref{lm:bijection}
the mappings $H$ and $H_A$ are bijective. 
In this case, $\row(X)=\R^p$ and $\calV^{\perp}=\{0\}$, 
which imply that the constraint \eqref{eq:rowspaceds} no longer exists.
Therefore, $|F|=p-|A|$, $T(\eta(s),A)$ is $p\times (p-|A|)$, and $M(r_A,s,A;\lambda)$ is $p\times p$.
Now,  the desired result is established by 
the same arguments as in Appendix~\ref{sec:proofhimden} with $U= X^\trans \veps/n$ in place of
$(\veps/\sqn)$ and $H_A$ in place of $\tdH_A$.

\subsection{Proof of Lemma~\ref{eq:MAparameterization}}\label{sec:pfparam}

First, the matrix $Q$ that defines the constraint \eqref{eq:rowspace} does not depend on $s$.
Second, the sphere $\mbS^{p_j-1}_{\alp^*}$, $j\in A$, can be parameterized by $s_{(j)\setminus k}$ 
for almost every point on the sphere, where $k\in \calG_j$ is chosen as the last component in the group.
More specifically, we may parameterize the positive half of $\mbS^{p_j-1}_{\alp^*}$ as
\begin{equation*}
\mbS^{p_j-1}_{\alp^*} \cap \left\{s_k > 0\right\} 
= \left\{\left(s_{(j)\setminus k}, [{1-\|s_{(j)\setminus k}\|_{\alp^*}^{\alp^*}}]^{1/\alp^*}\right): s_{(j)\setminus k} \in \mbB^{p_j-1}_{\alp^*}\right\},
\end{equation*}
and the negative half in a similar way, both using the variables indexed by $\calG_j\setminus k$.
Therefore, we can always choose $F(s,A)=F(A)$ to parameterize almost every point on $\calM_A$.

\subsection{Proof of Theorem~\ref{thm:densityinfty}}\label{sec:proofdeninf}

Put $R=Q^{\trans}$ with the matrix $Q$ as in \eqref{eq:rowspace}
and let the index set $I=K\cup \calG_{A^c}\subset \N_p$.
Any $s\in \calM_{A,K}$ must satisfy the following equality constraints:
$s_{K^c}=0$, $R_I s_I =0$, and $\|s_{K_j}\|_1=1$ for all $j\in A$,
which in turn impose constraints on its differentiation $ds$:
\begin{equation*}
ds_{K^c}=0,\quad\quad R_{I} ds_{I}  =0, 
\quad\text{and}\quad \la \sgn(s_{K_j}),d s_{K_j}\ra=0\;\forall j\in A.
\end{equation*}
Note that there are $p-n+|A|$ linearly independent constraints 
on $ds_I$ in the above under Assumption~\ref{as:X}.
Thus, $ds_I$ has $(n-|A|-|K^c|)=q$ free coordinates and 
there is a matrix $T=T(\sgn(s_K),A,K)$ so that
\begin{equation}\label{eq:dsI}
ds_I=T(\sgn(s_K),A,K) ds_F,
\end{equation} 
where $T \in \R^{(p-|K^c|) \times q}$ is a rank $q$ matrix and $F=F(\sgn(s_K),A,K)\subset\N_p$ with $|F|=q$.

Starting from \eqref{eq:KKTU}, we have
\begin{align*}
H(\hbeta,S;\beta_0,\lambda)&=\Psi \hbeta + \lambda {W} S - \Psi \beta_0 \\
&=\sum_{j\in\calA}\left[\hgamma_{j}\Psi_{\calK_j} \sgn(S_{\calK_j}) +\Psi_{\calK^c_j}\hbeta_{\calK^c_j}\right]
+ \lambda {W_\calI} S_{\calI} - \Psi \beta_0,
\end{align*}
where the index set $\calI=\calK \cup \calG_{\calA^c}$. Recalling the definition of the matrix $\Psi\circ \sgn(s)$,
we arrive at
\begin{equation*}
H(\hbeta,S)=[\Psi\circ \sgn(S)]_\calA \hgamma_{\calA} + \Psi_{\calK^c}\hbeta_{\calK^c}
+ \lambda {W_\calI} S_{\calI} - \Psi \beta_0,
\end{equation*}
where we have used the fact that $\sgn(S_{\calK_j^c})=0$ for $j\in\calA$.
Denote the value of $\hbeta$ by $b\in \R^p$. 
Fixing $\calA=A$ and $\calK=K$, the differentiation of $H_{A,K}$ at $(r_A,b_{K^c},s)\in \Omg_{A,K}$
is 
\begin{align*}
dH_{A,K}&=[\Psi\circ \sgn(s)]_A dr_{A} + \Psi_{K^c}d b_{K^c}
+ \lambda {W_I} ds_{I} \\
& = [\Psi\circ \sgn(s)]_A dr_{A} + \Psi_{K^c}d b_{K^c}+ \lambda {W_I}T ds_F=M(s,A,K;\lmd) d\theta,
\end{align*}
by plugging in \eqref{eq:dsI} for $ds_I$ and putting $\theta=(r_A,b_{K^c},s_F)\in \R^n$.
Then the Jacobian of the mapping $\tdH_{A,K}=\sqn(X^\trans)^+ H_{A,K}$ is given by \eqref{eq:jacobinfty}.
Similar to Lemma~\ref{lm:bijection}, $H_{A,K}$ and hence $\tdH_{A,K}$ are bijections on $\Omg_{A,K}$. 
Now following the proof of Theorem~\ref{thm:densityhigh} leads to 
the desired joint density.

\section{Technical Details in Examples}\label{sec:techforexps}

\subsection{Results and derivations in Example~\ref{exp:simple}}\label{sec:derivation1}

We present the complete results for Example~\ref{exp:simple} in this appendix.
The densities are given by the following differential forms:
\begin{align}
f_{\varnothing}ds_1ds_2&=\frac{\lmd^2}{\pi \sigma^2}\exp\left[-\frac{\lmd^2(s_1^2+s_2^2)}{\sigma^2}\right]
ds_1ds_2,\label{eq:exp1pdf0}\\
f_{\{1\}}dr_1 ds_1&= \frac{1}{\pi \sigma^2}\exp\left[-\frac{(r_1+\lmd)^2}{\sigma^2}\right]
\frac{r_1+\lmd}{|s_2|}dr_1ds_1 \label{eq:exp1pdf1copy} \\
f_{\{2\}}dr_2ds_1&=\frac{2\lmd}{\pi\sigma^2}\exp\left[-\frac{2r_2(r_2+\lmd)}{\sigma^2}\right]
\exp\left[-\frac{\lmd^2(s_1^2+s_2^2)}{\sigma^2}\right]dr_2ds_1,\label{eq:exp1pdf2}\\
f_{\{1,2\}}dr_1dr_2&=\frac{1}{\pi\sigma^2}\exp\left[-\frac{(r_1+r_2+\lmd)^2+r_2^2}{\sigma^2}\right]dr_1dr_2,
\label{eq:exp1pdf3}
\end{align} 
where $r_j>0$ for $j\in A$ and $s_{(1)}\in \mbD_A$ for $A=\varnothing, \{1\}, \{2\}, \{1,2\}$.
Special care must be taken when integrating over the parameter domains of these densities. For example,
\begin{equation*}
(r_1,r_2,s_{(1)}) \in (\R^+)^2 \times \{a,b,c,d\}:=\Theta \quad\text{ for }A=\{1,2\},
\end{equation*}
and therefore
\begin{equation*}
\Prob(\calA=\{1,2\})=\int_\Theta d\mu_{\{1,2\}} = 4 \int_{0}^{\infty} \int_{0}^{\infty} f_{\{1,2\}}dr_1dr_2.
\end{equation*}

Next, we derive these results and find $\Prob(\calA=A)$. A few pre-calculations are in order:
\begin{equation*}
\sqn(X^\trans)^+ =
\frac{1}{3}\left[\begin{array}{ccc}
2 & -1 & 1 \\
-1 & 2 & 1
\end{array}
\right], \quad\quad
\Psi =
\left[\begin{array}{ccc}
1 & 0 & 1 \\
0 & 1 & 1 \\
1 & 1 & 2
\end{array}
\right],
\end{equation*}
which lead to 
\begin{equation*}
r\circ\Psi =
\left[\begin{array}{ccc}
r_1 & 0 & r_2 \\
0 & r_1 & r_2 \\
r_1 & r_1 & 2r_2
\end{array}
\right], \quad\quad
\Psi \circ s=
\left[\begin{array}{cc}
s_1 &  s_3 \\
s_2 & s_3 \\
s_3 & 2s_3
\end{array}
\right].
\end{equation*}
The density function $g_n(v)=\phi_2(v; \sigma^2/2)$ and, with \eqref{eq:exp1rowspace},
\begin{align*}
\tdH_A(r_A,s)=\sqn(X^\trans)^+ \left[\sum_{j\in A} r_j (\Psi\circ s)_j + \lmd s \right] 
= (r_1+\lmd) s_{(1)}+r_2s_3 \I,
\end{align*}
where $\I=(1,1)$ is a (column) vector of ones. 
Constraint \eqref{eq:exp1rowspace} implies that
\begin{equation}\label{eq:exp1rowds}
ds_3=ds_1 + ds_2.
\end{equation}
We will first go through the calculations for $A=\{1\}$ and then move to the other three 
cases. In what follows, let $\tau=\sqrt{2}\lmd/\sigma$ and $Z=(Z_1,Z_2) \sim \dnorm_2(0,\bfI_2)$.

{\em Case 1:} $A=\{1\}$, $r_1>0$ and $r_2=0$.
Combining constraints \eqref{eq:exp1rowds} and $\|s_{(1)}\|^2=1$ leads to
\begin{equation*}
ds = \left[\begin{array}{c}
1 \\-s_1/s_2\\1-(s_1/s_2)
\end{array}\right] ds_1 \Rightarrow T(s,\{1\})=
\left[\begin{array}{c}
1 \\-s_1/s_2\\1-(s_1/s_2)
\end{array}\right].
\end{equation*}
Plugging in $r_2=0$ and $s_3=s_1+s_2$, it is then easy to verify that
\begin{equation*}
M(r_1,s,\{1\})=[s \mid (r_1+\lmd) T], \quad 
J_{\{1\}}(r_1,s)=-(r_1+\lmd)/s_2,\quad
\tdH_{\{1\}}(r_1,s)= (r_1+\lmd)s_{(1)}.
\end{equation*}
Consequently, we obtain the density as in \eqref{eq:exp1pdf1copy}.
The two arcs in $\mbD_{\{1\}}$ (its interior) are parameterized as
\begin{align*}
\partial(b,c)&=\left\{\left(s_1,({1-s_1^2})^{1/2}\right): -1< s_1< 0\right\} \\
\partial(d,a)&=\left\{\left(s_1,-({1-s_1^2})^{1/2}\right): 0< s_1< 1\right\}.
\end{align*}
It then follows that
\begin{align}
\Prob(\calA=\{1\})&=\int_{-1}^0\int_0^{\infty}f_{\{1\}}\left(r_1,s_1,({1-s_1^2})^{1/2}\right)dr_1 ds_1 \nonumber\\
&\quad\quad\quad\quad+\int_{0}^1\int_0^{\infty}f_{\{1\}}\left(r_1,s_1,-({1-s_1^2})^{1/2}\right)dr_1 ds_1\nonumber \\
&= \frac{1}{\pi \sigma^2}\left\{\int_0^{\infty}\exp\left[-\frac{(r_1+\lmd)^2}{\sigma^2}\right] (r_1+\lmd) dr_1\right\}
\left\{2\int_0^1\frac{1}{\sqrt{1-s_1^2}}ds_1\right\} \nonumber\\
&=\frac{1}{2}e^{-\lmd^2/\sigma^2}=  \frac{1}{2} \Prob(\|Z\|\geq \tau). \label{eq:exp1PrA1ofZ}
\end{align}

{\em Case 2:} $A=\varnothing$, $r_1=r_2=0$ and $s_{(1)}\in \mbD_{\varnothing}$.
By \eqref{eq:exp1rowds}, we have
\begin{equation*}
T(s,\varnothing)=
\left[\begin{array}{cc}
1 &0  \\0 & 1\\1& 1
\end{array}\right],
\end{equation*}
which in combination with $r_1=r_2=0$ leads to the following intermediate results:
\begin{equation*}
M(s,\varnothing)=\lmd T(s,\varnothing),\quad J_{\vn}(s)=\lmd^2,\quad
\tdH_{\vn}(s)=\lmd s_{(1)}.
\end{equation*}
Then the density $f_{\vn}$ is obtained immediately as in \eqref{eq:exp1pdf0} and
\begin{align}\label{eq:exp1PrA0}
\Prob(\calA=\vn) =\int_{\mbD_{\vn}} f_{\vn}(s_{(1)})ds_1ds_2 = \Prob(Z \in \tau \mbD_{\vn}).
\end{align}

{\em Case 3:} $A=\{2\}$, $r_1=0$ and $r_2>0$. The interior of $\mbD_{\{2\}}$ is parameterized as
\begin{align*}
\partial(a,b)&=\left\{\left(s_1,1-s_1\right): 0< s_1< 1\right\} \\
\partial(c,d)&=\left\{\left(s_1,-(1+s_1)\right): -1< s_1< 0\right\}.
\end{align*}
Since $s_3=s_1+s_2 \in \{1,-1\}$ in this case, we have $ds_3=ds_2+ds_1=0$ and thus
\begin{equation*}
T(s,\{2\})=
\left[\begin{array}{c}
1   \\- 1\\0
\end{array}\right]
\quad\text{ and }\quad
M(s,\{2\})=
\left[\begin{array}{cc}
s_3 & \lmd  \\s_3 & -\lmd \\2s_3 & 0
\end{array}\right],
\end{equation*}
using the fact that $r_1=0$. Now straightforward calculations give
\begin{equation*}
J_{\{2\}}(s)=-2\lmd s_3, \quad \tdH_{\{2\}}(r_2,s)=\lmd s_{(1)}+r_2s_3 \I.
\end{equation*}
Substituting $s_3$ by $s_1 + s_2$ with
the fact that $|s_3|=1$ leads to the density in \eqref{eq:exp1pdf2}.
Consequently,
\begin{align*}
\Prob(\calA=\{2\})= 
\frac{2\lmd}{\pi\sigma^2}\left\{\int_0^{\infty}\exp\left[-\frac{2r_2(r_2+\lmd)}{\sigma^2}\right] dr_2\right\}
\left\{2\int_{0}^1 \exp\left[-\frac{\lmd^2(s_1^2+(1-s_1)^2)}{\sigma^2}\right] ds_1\right\},
\end{align*}
utilizing the symmetry of $(s_1^2+s_2^2)$ between $\partial(a,b)$ and $\partial(c,d)$. The second integral
\begin{align*}
\int_{0}^1 \exp\left[-\frac{\lmd^2(s_1^2+(1-s_1)^2)}{\sigma^2}\right] ds_1 
=\frac{\sqrt{2\pi\sigma^2}}{2\lmd} \exp\left(-\frac{\lmd^2}{2\sigma^2}\right) \Prob(|Z_2|\leq\lmd/\sigma).
\end{align*}
After completing the first integral, we have
\begin{align}
\Prob(\calA=\{2\})&=2\cdot \Prob(Z_1\geq \lmd/\sigma\text{ and }|Z_2|\leq\lmd/\sigma) \nonumber\\
& = 2\cdot \Prob(Z_1+Z_2 \geq \tau \text{ and } |Z_1-Z_2|\leq \tau). \label{eq:exp1PrA2}
\end{align}

{\em Case 4:} $A=\{1,2\}$, $r_1,r_2>0$ and $\mbD_{\{1,2\}}=\{a,b,c,d\}$. Since $|A|=n=2$, 
\begin{equation*}
M(r,s,\{1,2\})=\Psi\circ s \quad\text{ and }\quad J_{\{1,2\}}(r,s)=s_1^2-s_2^2.
\end{equation*}
It is easy to see that $|J_{\{1,2\}}|=1$ for all $s_{(1)}\in \{a,b,c,d\}$ and
\begin{equation*}
\tdH_{\{1,2\}}(r,s)=(r_1+\lmd) s_{(1)}+r_2s_3 \I.
\end{equation*}
Then we obtain the density
$f_{\{1,2\}}$ in \eqref{eq:exp1pdf3} immediately, which leads to
\begin{align}
\Prob(\calA=\{1,2\}) & = 4 \int_{0}^{\infty} \int_{0}^{\infty} 
\frac{1}{\pi\sigma^2}\exp\left[-\frac{(r_1+r_2+\lmd)^2+r_2^2}{\sigma^2}\right]dr_1dr_2 \nonumber\\
& = 4\cdot \Prob(Z_1\geq 0 \text{ and } Z_2 -Z_1 \geq \tau). \label{eq:exp1PrA3}
\end{align}

Finally, by \eqref{eq:exp1PrA1ofZ}, \eqref{eq:exp1PrA0}, \eqref{eq:exp1PrA2}, and \eqref{eq:exp1PrA3} 
one can easily verify that
\begin{equation*}
\sum_A\Prob(\calA=A)=\Prob(Z \in \R^2)=1.
\end{equation*}

\subsection{Derivations in Example~\ref{exp:orthogonal}}\label{sec:derivation2}

This section is divided into three parts:

{\em Part 1:} Derivation of \eqref{eq:densityorth}.
For $v\in \R^p$ and $j\in\N_J$, define $u=v_{\langle j\rangle}\in\R^p$ so that
$u_k=v_k$ for $k\in\calG_j$ and $u_k=0$ otherwise.
Let $b\in\R^p$ denote the value for $\hbeta$, i.e. $b_{(j)}=r_js_{(j)}$ for $j\in \N_J$.
Straightforward algebra leads to:
\begin{align*}
H_A(r_A,s)& =b+\lambda \sqm s-\beta_0, \\
\Psi\circ s & = \left[s_{\langle 1\rangle}|\ldots| s_{\langle J\rangle}\right] \in \R^{p \times J}, \\
r\circ\Psi + \lambda W & = \diag\left\{(r_j+\lambda\sqm)\bfI_m: j\in \N_J\right\} \in \R^{p\times p}.
\end{align*}
Since $\row(X)=\R^p$, the constraint \eqref{eq:rowspaceds} disappears.
For $j\in A$, choose $k(j)\in\calG_j$ such that $s_{k(j)}\ne 0$ and put $F(j)=\calG_j\setminus\{k(j)\}$. 
Then constraint \eqref{eq:sphereds} can be written as
\begin{equation*}
ds_{k(j)}=-\frac{1}{s_{k(j)}} \langle s_{F(j)}, d s_{F(j)}\rangle \quad \text{ for }j\in A.
\end{equation*}
Without loss of generality, assume that $k(j)=m\cdot j$ to be the last component in the group.
The matrix $T=T(s,A)$ has a blockwise diagonal structure and its $j^\supth$ block
\begin{align*}
T(j) = \left[
\begin{array}{c}
\bfI_{m-1} \\ -s_{F(j)}^\trans /s_{k(j)}
\end{array}
 \right]  \text{ for } j \in A  \quad \text{ and } \quad
T(j)= \bfI_m  \text{ for } j\notin A.
\end{align*}
It follows immediately that 
\begin{equation*}
(r\circ\Psi + \lambda W)T(s,A)=\diag\left\{(r_j+\lambda\sqm)T(j): j\in \N_J\right\}.
\end{equation*}
Permuting the columns of $M$ \eqref{eq:defM} to put $s_{\langle j \rangle}$, $j\in A$, 
to the right of the $j^\supth$ block of
the above matrix, $M$ is also seen to be blockwise diagonal with each block $M_{(jj)}$ of size $m\times m$.
For $j\in A$, the $j^\supth$ block 
\begin{align*}
M_{(jj)} = \left[(r_j+\lambda\sqrt{m}) T(j) \mid s_{(j)} \right],  
\end{align*}
and for $j\notin A$, since $r_j=0$,
\begin{equation*}
M_{(jj)}=(\lambda\sqm)\bfI_m.
\end{equation*} 
Simple calculation with $\|s_{(j)}\|^2=1$ for $j\in A$ shows that 
\begin{equation}\label{eq:detMjjexp2}
|\det M_{(jj)}|=\left\{
\begin{array}{ll}
{(r_j+\lambda \sqm)^{m-1}}/{|s_{k(j)}|} & j \in A \\
(\lambda \sqm)^{m} & j\notin A.
\end{array}
\right.
\end{equation}
Under the hypotheses, $\sqn(X^{\trans})^+=X/\sqn$ is an orthogonal matrix whose determinant is
$\pm 1$. Consequently, the Jacobian \eqref{eq:jacob} is
\begin{equation}\label{eq:jacoborth}
|J_A|=\prod_{j=1}^J |\det M_{(jj)}|.
\end{equation}
Plugging \eqref{eq:jacoborth} into \eqref{eq:dGaussian} with $\Psi=\bfI_p$, 
we obtain the distribution given in \eqref{eq:densityorth}.

{\em Part 2:}
Derivation of the marginal density of $\hgamma_j$ \eqref{eq:densityrjorth} for $j\in A$. 
Denote by $\mu_j$ the joint distribution of $\hgamma_j$ and $S_{(j)}$.
We start from the integral 
\begin{align*}
f_j(r_j)dr_j = \int_{\mbS^{m-1}}d\mu_j 
=C(m) ({2\pi \sigma^2}/{n})^{-\frac{m}{2}} (r_j + \lambda \sqm)^{m-1} 
\exp\left[ -\frac{n}{2\sigma^2}(r_j+\lambda \sqm)^2\right] dr_j, 
\end{align*}
where $C(m)>0$ is a constant:
\begin{align*}
C(m)&=\int_{\mbS^{m-1}} \frac{1}{|s_{k(j)}|} ds_{F(j)}
=2\int_{\mbB^{m-1}} {\left(1-\|v\|^2\right)^{-1/2}} dv.  
\end{align*}
With a change of variable, $v=x/\sqrt{1+\|x\|^2}$,
\begin{align*}
C(m)=2\int_{\R^{m-1}} {\left(1+\|x\|^2\right)^{-m/2}} dx = \frac{2\cdot \pi^{m/2}}{\Gamma(m/2)},
\end{align*}
by the normalizing
constant of the multivariate $t$-distribution with one degree of freedom.

{\em Part 3:}
Proof of the last identity in \eqref{eq:gmsoft}.
It follows from the sampling distribution of $\tdbeta_{(j)}$ that 
$(n/\sigma^2)\|\tdbeta_{(j)}\|^2\eqinL \chi^2_m$, following a $\chi^2$-distribution
with $m$ degrees of freedom.
Letting $z=(r_j+\lambda\sqm)^2$, we have
\begin{align*}
\int_t^{\infty} f_j (r_j) d r_j&=\int_{(t+\lmd\sqm)^2}^{\infty} 
\frac{\left({n}/{\sigma^2}\right)^{\frac{m}2}}{2^{{m}/{2}} \cdot\Gamma(m/2)}
z^{{m}/{2}-1} \exp\left(-\frac{n}{2\sigma^2}z\right)  dz \\
& = \Prob\left\{(\sigma^2/n)\chi^2_m > (t+\lmd \sqm)^2\right\} \\
& =\Prob\left\{\|\tdbeta_{(j)}\|^2 > (t+\lmd \sqm)^2\right\},
\end{align*} 
which completes the proof.

\subsection{Derivations in Example~\ref{exp:lasso}}\label{sec:derivation3}

We note that the constraint \eqref{eq:sphereds} reduces to $ds_j=0$ for $j\in A$, 
which implies that $T_{A\bullet}=\bfzr$ as in property (ii). Recall that $B=\N_p\setminus A$.
The constraint imposed on $d s_B$
comes from \eqref{eq:rowspaceds} and is thus independent of $s$, hence property (i). As a consequence,
the set of free coordinates of $s$ is always a subset of $B$, i.e. $F\subset B$ and $|F|=n-|A|$.
Since $r_{B}=0$ by definition, $(r\circ\Psi)_j=0$ for all $j\in B$.
It follows that $(r\circ\Psi)T=\bfzr$ and thus as defined in \eqref{eq:defM}
\begin{equation*}
M(r_A,s,A)  = \left[(\Psi\circ s)_A \mid \lambda W T\right]
=\left[(\Psi\circ s)_A \mid \lambda W_{B} T_{B\bullet}\right].
\end{equation*}
Since $|s_j|=1$ for $j\in A$ and $p_j=1$,
\begin{align*}
|J_A(r_A,s)|&=\left|\det \left\{\sqn (X^\trans)^+ [(\Psi\circ s)_A \mid \lambda W_{B} T_{B\bullet}]\right\}\right| \\
&=\left|\det \left\{\sqn (X^\trans)^+ [\Psi_A \mid \lambda W_{B} T_{B\bullet}]\right\}\right|.
\end{align*}
Substituting this into \eqref{eq:nform} gives the density in \eqref{eq:nformLasso}.

To compare \eqref{eq:nformLasso} with Theorem 2 of \cite{Zhou14}, we apply the following change of
variable: Let $b_j=s_jr_j$ denote the value for $\hbeta_j$ for $j\in A$.
Plugging into \eqref{eq:nformLasso} that $r_j=|b_j|$, $s_j=\sgn(b_j)$ and $dr_j=s_j d b_j$ for $j\in A$, 
we obtain the density for $(\hbeta_\calA,S_\calB,\calA)$ parameterized by $(b_A,s_F)$:
\begin{equation}\label{eq:nformLassoalter}
g_n(\tdH_A(|b_A|,s))
\left|\det \left\{\sqn (X^\trans)^+ [\Psi_A \mid \lambda W_{B} T_{B\bullet}]\right\}\right| db_A ds_F,
\end{equation}
where we have again used $|s_j|=1$ for $j\in A$ in the change of the volume elements.

\bibliographystyle{asa}
\bibliography{sparsereferences}

\end{document}